\documentclass[reqno,11pt]{amsart}
\usepackage[margin=1in]{geometry}
\usepackage{amsmath,amssymb,amsthm,amsfonts}
\usepackage{graphicx}
\usepackage{booktabs}
\usepackage{mathtools}
\usepackage{enumitem}
\usepackage{etoolbox}
\usepackage{hyperref}
\usepackage{multirow}

\newcommand{\p}{\partial}
\newcommand{\R}{\mathbb{R}}

\newcommand{\mrw}{\mathring{w}}
\newcommand{\mrz}{\mathring{z}}
\newcommand{\mrs}{\mathring{s}}

\newcommand{\clip}{\operatorname{clip}_{\pm}}

\theoremstyle{plain}
\newtheorem{theorem}{Theorem}[section]
\newtheorem{lemma}[theorem]{Lemma}
\newtheorem{proposition}[theorem]{Proposition}

\newtheorem{algorithm}{Algorithm}[section]
\newtheorem{remark}{Remark}[section]

\theoremstyle{definition}
\newtheorem{definition}[theorem]{Definition}

\numberwithin{equation}{section}

\usepackage{etoolbox} 

\makeatletter
\patchcmd{\@settitle}{\uppercasenonmath\@title}{}{}{}

\def\@settitle{
  \begin{center}
    \baselineskip14\p@\relax
    \large\bfseries
    \@title
  \end{center}
}
\makeatother

\makeatletter
\def\@setauthors{
  \begingroup
  \trivlist
  \centering\footnotesize \@topsep30\p@\relax
  \advance\@topsep by -\baselineskip
  \item\relax
  \andify\authors
  \def\baselinestretch{1}
  \authors
  \endtrivlist
  \endgroup
}
\makeatother

\title[Sub-cell Wave Reconstruction from DRVs]{Sub-cell Wave Reconstruction from Differentiated Riemann Variables}
\author[Steve Shkoller]{\large Steve Shkoller}

\begin{document}

\begin{abstract}
We introduce a postprocessing procedure that recovers sub-cell wave 
geometry from a standard one-dimensional Euler shock-capturing 
computation using differentiated Riemann variables (DRVs)---characteristic 
derivatives that separate the three wave families into distinct 
localized spikes.  Filtered DRV surrogates detect the waves, 
plateau sampling extracts the local states, and a pressure-wave-function 
Newton closure completes the geometry.  The entire pipeline adds less 
than $0.25\%$ to the cost of a baseline WENO--5/HLLC solve. 
For Sod, a severe-expansion problem, and the LeBlanc shock tube, 
wave locations are recovered to within roundoff or $O(10^{-4})$ 
and the contact is sharpened to one cell width; a pattern-agnostic 
extension handles all four Riemann configurations with errors at the 
$10^{-6}$--$10^{-8}$ level.  Direct comparison with MUSCL--THINC--BVD 
and WENO-Z--THINC--BVD shows that neither reproduces the combination 
of sharp contacts, small contact-window internal-energy error, 
and elimination of the LeBlanc positive overshoot achieved by the 
DRV reconstruction.
\end{abstract}

\maketitle

\setcounter{tocdepth}{1}
\tableofcontents

\section{Introduction}
\label{sec:intro}

High-resolution Godunov-type shock-capturing schemes approximate discontinuous Euler flows on a fixed Eulerian mesh by replacing each jump with a transition layer spread over several cells \cite{Godunov1959,Harten1983}. For mild problems this smearing is often acceptable, but for strong one-dimensional Riemann problems it can obscure the sub-cell wave geometry and can also create visible thermodynamic defects near the contact discontinuity. The severe-expansion and LeBlanc tests are especially unforgiving in this respect \cite{TangLiu2006}: a numerically broadened contact may be accompanied by a conspicuous error in the specific internal energy. The question we address  is whether the lost wave geometry can be recovered \emph{a posteriori} from an already computed conservative solution, at negligible additional cost, and with enough accuracy to remove those thermodynamic defects.

The key observation is that differentiated state variables (i.e., velocity, density,  energy) are not the right objects for this task. For the one-dimensional non-isentropic Euler equations, one should first diagonalize the quasilinear system and only then differentiate. This produces three differentiated characteristic variables associated with the left acoustic wave, the contact discontinuity, and the right acoustic wave. Following the terminology we introduced in \cite{BuckmasterDrivasShkollerVicol2022,ShkollerVicol2024}, these are called differentiated Riemann variables (DRVs). In smooth regions they reduce to explicit algebraic combinations of $\p_x u$, $\p_x p$, and $\p_x s$; at discontinuities they appear numerically as localized spikes approximating the singular Radon measures carried by the three wave families. This family separation is exactly what makes them useful for reconstruction: the $1$-wave, the contact, and the $3$-wave are detected in different variables rather than by a single generic edge sensor.

The reconstruction algorithm is correspondingly simple. A standard finite-volume shock-capturing method evolves the conservative Euler variables. At final time, algebraic DRV surrogates are extracted from the primitive fields, adaptively filtered, and used to infer an initial $1$--$2$--$3$ wave geometry. The four plateau states (far-left, left-star, right-star, far-right) are then sampled from the baseline solution, and the local star region is closed by iterating the classical pressure-wave-function equation
\[
F(p)=f_L(p)+f_R(p)+u_R-u_L
\]
from a sampled-pressure seed. Because the method operates on a single final-time snapshot, the entire postprocessing pipeline---DRV surrogates, filtering, spike detection, plateau sampling, Newton closure, and piecewise reconstruction---adds less than $0.25\%$ to the cost of the baseline solve, while producing qualitative improvements in wave-location accuracy and contact sharpness.

The results are striking. For the three main benchmarks (Sod, a severe-expansion problem, and the LeBlanc shock tube), the reconstruction places wave locations to within roundoff or to the $10^{-4}$ level, sharpens the contact to essentially one cell width, and eliminates the spurious contact-layer internal-energy overshoot that persists in the baseline WENO--5/HLLC solution. Direct comparison with two strong jump-like alternatives---MUSCL--THINC--BVD and WENO-Z--THINC--BVD---shows that neither reproduces the combination of near-discontinuous contacts, small contact-window internal-energy error, and LeBlanc overshoot suppression achieved by the DRV method. A pattern-agnostic extension with adaptive Newton iteration and a residual-based convergence guard handles all four single-interface Riemann configurations, including a near-vacuum two-rarefaction test and a two-shock collision, with wave-location errors at the $10^{-6}$--$10^{-8}$ level.

The paper has two intertwined aims. The mathematical aim is to explain why DRVs are the correct spike-detection variables: they arise from diagonalizing before differentiation, they admit a conservative derivative form, and at a pure contact the acoustic source has no singular contribution. The numerical aim is to show that this information can be converted into a practical and unusually effective reconstruction procedure with negligible computational overhead.

The paper is organized as follows. Section~\ref{sec:scalar} isolates the scalar first-moment mechanism that underlies spike-based wave localization. Section~\ref{sec:why-drvs} explains why DRVs are the correct spike-detection variables for the Euler system and proves that the algebraic surrogates used in the code are exact continuum identities in smooth regions. Section~\ref{sec:strategy} defines the Riemann-solver terminology used throughout, states the reconstruction algorithm, and records the precise formulas needed to reproduce it; it concludes with a comparison to existing reconstruction schemes. Section~\ref{sec:contact-prop} explains why the reconstruction structurally eliminates the contact-layer internal-energy overshoot. Section~\ref{sec:results} reports wave-location errors, conservation defects, a grid study, the three benchmark reconstructions, and the independent comparison study. Section~\ref{sec:generalized} extends the method to all four single-interface Riemann configurations and reports computational costs. Appendix~\ref{app:implementation} records the baseline finite-volume solver, and Appendix~\ref{app:exact-closure} records the optional exact local Riemann closure.

\section{Differentiated fields locate waves: a scalar mechanism}
\label{sec:scalar}

The DRV reconstruction locates each wave by finding a localized spike in the appropriate differentiated characteristic variable and computing its center of mass. Before developing the full Euler theory in Section~\ref{sec:why-drvs}, we isolate the mechanism in the simplest possible setting: a single viscously regularized shock in a scalar conservation law. The point is that the differentiated field $v=\partial_x u$ forms a localized spike whose center of mass propagates at exactly the Rankine--Hugoniot speed. This first-moment property is what makes differentiated fields the right objects for wave localization.

Consider a strictly convex scalar conservation law
\[
\partial_t u + \partial_x f(u)=0,
\]
and let $v=\partial_x u$. For a viscously regularized isolated shock layer,
\begin{equation}
\partial_t v + \partial_x\!\big(f'(u)v\big)=\nu\partial_{xx}v,
\label{eq:scalar-v}
\end{equation}
with $\nu>0$, the differentiated field $v$ is a localized smooth approximation of $[u]\delta_{x_s(t)}$.

For an interval $[x_L,x_R]$ containing the regularized layer, define the mass and center of mass of $v$ by
\[
M(t)=\int_{x_L}^{x_R} v(x,t)\,dx,
\qquad
X_{cm}(t)=\tfrac{\int_{x_L}^{x_R} x\,v(x,t)\,dx}{\int_{x_L}^{x_R} v(x,t)\,dx}.
\]

\begin{lemma}[First-moment transport for a regularized differentiated shock layer]
\label{lem:scalar-moment}
Assume that $[x_L,x_R]$ contains a single viscously regularized shock layer connecting constant endstates $u_L$ and $u_R$, and that $v$, $v_x$, and $f'(u)v$ are exponentially small at $x=x_L,x_R$. Then
\[
M(t)=[u]+O(e^{-c/\nu}),
\qquad
\tfrac{d}{dt}X_{cm}(t)=\tfrac{[f(u)]}{[u]}+O(e^{-c/\nu}).
\]
In particular, the center of mass of the differentiated layer propagates at the Rankine--Hugoniot speed up to exponentially small endpoint leakage.
\end{lemma}

\begin{proof}
Integrating \eqref{eq:scalar-v} over $[x_L,x_R]$ shows that $\dot M(t)$ is an endpoint term. Since $v=\partial_x u$, one has $M(t)=u(x_R,t)-u(x_L,t)=[u]+O(e^{-c/\nu})$. Next,
\[
\tfrac{d}{dt}\int_{x_L}^{x_R}xv\,dx
=
\int_{x_L}^{x_R}x\bigl(-\partial_x(f'(u)v)+\nu\partial_{xx}v\bigr)\,dx.
\]
After one integration by parts in each term, all boundary contributions are exponentially small, while
\[
\int_{x_L}^{x_R}f'(u)v\,dx
=
\int_{x_L}^{x_R}\partial_x f(u)\,dx
=
[f(u)] + O(e^{-c/\nu}).
\]
The quotient rule therefore yields
\[
\tfrac{d}{dt}X_{cm}(t)=\tfrac{[f(u)]}{[u]}+O(e^{-c/\nu}).
\]
\end{proof}

Lemma~\ref{lem:scalar-moment} is the basic reason why the DRV reconstruction works. A numerically captured shock or contact produces a localized spike in the appropriate differentiated field. The mass of that spike encodes the jump magnitude, while its center of mass encodes the wave position. The DRV algorithm (Algorithm~\ref{alg:recon}, Step~4) exploits exactly this structure: it computes the center of mass of the filtered DRV spikes to localize each wave family. The task of Section~\ref{sec:why-drvs} is to identify the differentiated variables---the DRVs---for which this first-moment mechanism extends cleanly to the three wave families of the non-isentropic Euler equations.

\section{DRVs are the correct spike-detection variables}
\label{sec:why-drvs}

\subsection{Symmetric Euler variables and exact DRVs}

For the one-dimensional non-isentropic Euler equations, it is convenient to work with the symmetric variables $(u,\sigma,s)$,
\begin{subequations}
\label{eq:sym}
\begin{align}
    \partial_t u + u\partial_x u + \alpha\sigma\partial_x\sigma &= \tfrac{\alpha}{2\gamma}\sigma^2\partial_x s, \\
    \partial_t \sigma + u\partial_x\sigma + \alpha\sigma\partial_x u &= 0, \\
    \partial_t s + u\partial_x s &= 0,
\end{align}
\end{subequations}
where $\alpha=\tfrac{\gamma-1}{2}$ and $\sigma=c/\alpha$ with $c=\sqrt{\gamma p/\rho}$. We introduce the classical Riemann variables
\[
w=u+\sigma,
\qquad
z=u-\sigma.
\]
Letting $q=(u,\sigma,s)^T$, the system can be written as
\[
\partial_t q + A(q)\partial_x q = 0,
\qquad
A(q)=
\begin{pmatrix}
 u & \alpha\sigma & -\tfrac{\alpha\sigma^2}{2\gamma}\\
 \alpha\sigma & u & 0\\
 0 & 0 & u
\end{pmatrix}.
\]
The wave speeds are
\[
\lambda_3=u+\alpha\sigma,
\qquad
\lambda_2=u,
\qquad
\lambda_1=u-\alpha\sigma,
\]
and a corresponding left-eigenvector matrix is
\[
L(q)=
\begin{pmatrix}
1&1&-\tfrac{\sigma}{2\gamma}\\
1&-1&\tfrac{\sigma}{2\gamma}\\
0&0&1
\end{pmatrix}.
\]
The differentiated state variables are $\partial_x q=(\p_x u,\sigma_x,\p_x s)^T$. The Eulerian differentiated Riemann variables are defined by projecting \emph{before} differentiation becomes singular:
\begin{subequations}
\label{eq:drv-defs}
\begin{align}
\mrw &:= \partial_x w - \tfrac{\sigma}{2\gamma}\partial_x s, \\
\mrz &:= \partial_x z + \tfrac{\sigma}{2\gamma}\partial_x s, \\
\mrs &:= \partial_x s.
\end{align}
\end{subequations}
Equivalently, $(\mrw,\mrz,\mrs)^T=L(q)\partial_x q$.

\subsection{Diagonalize first and  differentiate second}

If one differentiates the state variables first, then one obtains the differentiated state-variable system
\begin{equation}
\partial_t(\partial_x q)+A(q)\partial_x(\partial_x q)=-(\partial_x q)^T\tfrac{\partial A}{\partial q}\,\partial_x q.
\label{eq:dsv-raw}
\end{equation}
The formal temptation is then to multiply by $L(q)$ and treat the result as a characteristic decomposition. At discontinuities, however, this procedure is not canonically defined in the classical distribution space.

\begin{proposition}[Obstruction to post-differentiation diagonalization]
\label{prop:dsv-failure}
Let $q$ be piecewise $C^1$ with a contact discontinuity. Then the formal characteristic decomposition obtained by differentiating first and diagonalizing afterwards is not canonically defined in classical distribution theory: the coefficients $L(q)$ are discontinuous on the jump set, whereas $\partial_x q$ contains singular measures, and the differentiated equations generate products that require choosing representatives of discontinuous coefficients on the support of those measures. In particular, terms of the type $H\delta$, $H\delta'$, and $\delta\cdot\delta$ arise formally.
\end{proposition}

\begin{proof}
Near a contact at $x=x_c$, write locally
\[
q(x)=q^-+[q]H(x-x_c)+r(x),
\]
with $r\in C^1$. Then
\[
q_x=[q]\delta_{x_c}+r_x,
\qquad
q_{xx}=[q]\delta'_{x_c}+r_{xx}.
\]
Since $L=L(q)$ depends smoothly on $q$, the jump in $q$ induces a jump in $L$:
\[
L(q)=L^-+[L]H(x-x_c)+\widetilde r(x).
\]
If one first differentiates and then multiplies by $L(q)$, the advective term contains factors such as $L(q)A(q)q_{xx}$, hence formally terms of type $H\delta'$. Likewise, expanding $\partial_x(L(q)q_x)$ produces $L_x(q)q_x$, and because $L_x$ carries a Dirac mass while $q_x$ already carries a Dirac mass, one formally obtains $\delta\cdot\delta$. These products are not canonically defined in the Schwartz distribution space \cite{Schwartz1954}. The pre-diagonalized DRV variables \eqref{eq:drv-defs} avoid this ambiguity.
\end{proof}

\subsection{Conservative derivative form and local contact orthogonality}

For the DRVs, the exact Eulerian differentiated system can be rewritten in conservative derivative form as
\begin{subequations}
\label{eq:drv-cons}
\begin{align}
\partial_t \mrw + \partial_x(\lambda_3\mrw) &= \mathcal S(\mrw,\mrz,\mrs), \\
\partial_t \mrz + \partial_x(\lambda_1\mrz) &= -\mathcal S(\mrw,\mrz,\mrs), \\
\partial_t \mrs + \partial_x(u\mrs) &= 0,
\end{align}
\end{subequations}
where
\begin{equation}
\mathcal S(\mrw,\mrz,\mrs)=\tfrac{\alpha}{4\gamma}\sigma(\mrw+\mrz)\mrs.
\label{eq:drv-source}
\end{equation}
Because $\mrw+\mrz=2\p_x u$ by \eqref{eq:drv-defs}, the source term may equivalently be written as
\[
\mathcal S=\tfrac{\alpha}{2\gamma}\sigma\,\p_x u \p_x s.
\]
Thus $\mrw$ is the 3-family differentiated field, $\mrz$ is the 1-family differentiated field, and $\mrs$ is the 2-family differentiated field; the only coupling is through the source term \eqref{eq:drv-source}.

At a pure contact of a self-similar Riemann solution, $u$ and $p$ are constant across the jump while $s$ jumps. Therefore the entropy differentiated field can produce a contact spike without automatically creating an acoustic spike there.

\begin{proposition}[Local orthogonality at a pure contact]
\label{prop:contact-orth}
Consider a self-similar Riemann solution in a neighborhood of an isolated contact discontinuity, separated from acoustic waves. Then $u$ and $p$ are constant in that neighborhood, while $s$ has a jump. Hence
\[
\partial_x u=0,
\qquad
\partial_x p=0,
\qquad
\partial_x s=[s]\delta_{x_c}
\]
in the distributional sense near the contact. Consequently the source term \eqref{eq:drv-source} has no singular contribution supported at the contact, and the contact singularity does not generate a singular acoustic source in $\mrw$ or $\mrz$.
\end{proposition}

\begin{proof}
For a self-similar Riemann solution $U(x,t)=\widetilde U(x/t)$, an isolated contact is represented by two constant intermediate states $U_*^-$ and $U_*^+$ separated at the contact speed $u_*$, with $u_*^-=u_*^+$ and $p_*^-=p_*^+$. Therefore $u$ and $p$ are constant on both sides with equal left and right values, so they are distributionally constant near the contact and their spatial derivatives vanish there. By contrast, $s$ is piecewise constant with a jump, hence $\partial_x s=[s]\delta_{x_c}$. Since $\partial_x u=0$ in that neighborhood, the product form $\mathcal S=\tfrac{\alpha}{2\gamma}\sigma\,\p_x u \p_x s$ has no singular contribution supported at the contact.
\end{proof}

\subsection{Exact algebraic identities and family selectivity}

The code does not evolve the DRV PDEs \eqref{eq:drv-cons}; it extracts final-time algebraic surrogates from the primitive fields. The important point is that these algebraic surrogates are not heuristic. They are exact continuum identities in smooth regions.

\begin{lemma}[Exact algebraic form of the smooth DRVs]
\label{lem:alg-identity}
Assume $\rho>0$ and $p>0$ are smooth. Then the exact DRVs satisfy
\begin{subequations}
\label{eq:alg-drvs-cont}
\begin{align}
\mrw &= \p_x u + \tfrac{\alpha\sigma}{\gamma p}p_x, \\
\mrz &= \p_x u - \tfrac{\alpha\sigma}{\gamma p}p_x, \\
\mrs &= \p_x s.
\end{align}
\end{subequations}
Hence the code formulas are obtained simply by replacing spatial derivatives in \eqref{eq:alg-drvs-cont} by centered differences.
\end{lemma}

\begin{proof}
From
\[
s=\log p - \gamma \log \rho,
\qquad
\sigma=\alpha^{-1}\sqrt{\gamma p/\rho},
\]
we obtain
\[
\tfrac{\sigma_x}{\sigma}
=
\tfrac12\tfrac{p_x}{p}-\tfrac12\tfrac{\rho_x}{\rho}
=
\tfrac12\tfrac{p_x}{p}-\tfrac{1}{2\gamma}\Bigl(\tfrac{p_x}{p}-\p_x s\Bigr)
=
\tfrac{\alpha}{\gamma}\tfrac{p_x}{p}+\tfrac{1}{2\gamma}\p_x s.
\]
Therefore
\[
\mrw=\p_x u+\sigma_x-\tfrac{\sigma}{2\gamma}\p_x s=\p_x u+\tfrac{\alpha\sigma}{\gamma p}p_x,
\]
and similarly
\[
\mrz=\p_x u-\sigma_x+\tfrac{\sigma}{2\gamma}\p_x s=\p_x u-\tfrac{\alpha\sigma}{\gamma p}p_x.
\]
The identity $\mrs=\p_x s$ is immediate from the definition.
\end{proof}

Lemma~\ref{lem:alg-identity} already explains why the implemented formulas are the correct ones. The next proposition explains why they are also family-selective.

\begin{proposition}[Family selectivity for the 1-wave and 2-wave detectors]
\label{prop:family-selectivity}
Consider an exact self-similar 1-rarefaction/2-contact/3-shock Riemann pattern in regions where the waves are separated.
\begin{enumerate}[label=\textnormal{(\roman*)},leftmargin=*]
    \item In the interior of the left 1-rarefaction, $\p_x s=0$ and $w$ is constant. Hence $\mrs=0$, $\mrw=0$, and $\mrz=z_x$.
    \item Near the isolated contact, $\mrw=0$ and $\mrz=0$ in the distributional sense, while $\mrs=[s]\delta_{x_c}$.
\end{enumerate}
Thus the rarefaction is naturally detected by $\mrz$, whereas the contact is naturally detected by $\mrs$.
\end{proposition}

\begin{proof}
In a smooth 1-simple wave, the entropy is constant and the 1-family Riemann invariant $w=u+\sigma$ is constant. Hence $\p_x s=0$ and $w_x=0$, so by \eqref{eq:drv-defs} one has $\mrs=0$ and $\mrw=0$, while $\mrz=z_x$. Part~(ii) is exactly Proposition~\ref{prop:contact-orth}.
\end{proof}

The remaining 3-wave is the right-going acoustic family. Since $\mrw$ is the differentiated field transported by the 3-family equation in \eqref{eq:drv-cons}, it is the family-adapted shock detector. In the implemented shock-tube experiments, the filtered $\mrw$ field exhibits a dominant negative compression spike, and the code uses that dominant negative spike to locate the shock.

\begin{table}[t]
\centering
\small
\resizebox{\textwidth}{!}{%
\begin{tabular}{llll}
\toprule
wave feature & local structure & active detector & implemented marker \\
\midrule
1-rarefaction & $\p_x s=0$, $w_x=0$, $z_x\neq0$ & $\mrz$ & extrema of $D_0(G\mrz)$ \\
2-contact & $\p_x u=p_x=0$, $\p_x s=[s]\delta$ & $\mrs$ & dominant spike of $|G\mrs|$ \\
3-shock / compression & right-going acoustic concentration & $\mrw$ & dominant negative spike of $G\mrw$ \\
\bottomrule
\end{tabular}}
\caption{Why the three algebraic DRV surrogates are the correct spike detectors for the implemented 1--2--3 wave pattern.}
\label{tab:detectors}
\end{table}

\section{The DRV-guided reconstruction strategy}
\label{sec:strategy}

The implemented method has two levels. The conservative Euler solve determines the macroscopic weak solution. The DRV fields then extract a family-separated geometric skeleton of the $1$--$2$--$3$ wave pattern. The primitive reconstruction is built from that geometric skeleton rather than directly from the smeared conservative profile.

\subsection{Terminology}
Before stating the algorithm, we fix the terminology and notation from the classical theory of the Riemann problem \cite{Toro2009} that the reconstruction relies on.

\begin{definition}[Star region and star states] The exact self-similar solution of a one-dimensional ideal-gas Riemann problem produces three waves emanating from the initial discontinuity: a left-going acoustic wave (the $1$-wave, either a rarefaction fan or a shock), a contact discontinuity (the $2$-wave), and a right-going acoustic wave (the $3$-wave). Between the $1$-wave and the $3$-wave lies the \emph{star region}, in which the velocity $u$ and the pressure $p$ take common values $u_*$ and $p_*$ on both sides of the contact, while the density $\rho$ and the specific entropy $s$ jump across it. The pair $(u_*,p_*)$, together with the left and right star densities $\rho_{*L}$ and $\rho_{*R}$, constitutes the \emph{star state}.
\end{definition} 

\begin{definition}[Plateaus and plateau sampling] In a numerically computed Euler solution, each constant-state region of the exact solution appears as an approximately flat region over several mesh cells. We call these flat regions \emph{plateaus}. There are four of them, separated by the three waves: the far-left plateau (carrying the original left data), the left-star plateau (between the $1$-wave and the contact), the right-star plateau (between the contact and the $3$-wave), and the far-right plateau (carrying the original right data). \emph{Plateau sampling} refers to the extraction of representative primitive-variable values $(\rho,u,p)$ from each plateau by taking trimmed medians over the interior cells of the corresponding interval.
\end{definition} 

\begin{definition}[Pressure-wave-function closure] The classical exact Riemann solver determines $p_*$ as the root of the \emph{pressure-wave-function equation}
\begin{equation}
\label{eq:pressure-wave-function}
F(p)=f_L(p)+f_R(p)+u_R-u_L = 0,
\end{equation}
where $u_L$ and $u_R$ are the sampled far-left and far-right velocities, and $f_L$ and $f_R$ are the left and right \emph{pressure functions}. Each $f_k(p)$ encodes the velocity change across the $k$-side acoustic wave as a function of a trial star pressure $p$: it takes the rarefaction form when $p\le p_k$ and the shock form when $p>p_k$ (the precise formulas are recorded in Section~\ref{sec:plateau-closure} below). In the present algorithm, $u_L$, $u_R$, $p_L$, and $p_R$ are not the original initial data; they are the \emph{sampled} far-left and far-right plateau values extracted from the DRV-guided geometry.
\end{definition} 

\subsection{DRV reconstruction algorithm}
At the conceptual level, the procedure is the following.

\begin{algorithm}[DRV-guided postprocessing and reconstruction]
\label{alg:recon}
\leavevmode\\
Given benchmark data $(\gamma,N,t_{\mathrm{final}},x_*,\text{left data},\text{right data})$, perform the following steps.
\begin{enumerate}[label=\textsf{Step \arabic*.},leftmargin=*]
    \item Solve the conservative Euler equations on the uniform mesh with the baseline primitive-variable WENO--5 / HLLC / SSP--RK3 method recorded in Appendix~\ref{app:implementation}.
    \item Convert the final conservative state to primitive variables, compute the algebraic DRV surrogates, and build the adaptive wave sensor.
    \item Apply the adaptive Gaussian filter to $\mrw$, $\mrz$, $\mrs$, and the positive velocity-gradient detector.
    \item Detect the initial shock, contact, rarefaction head, and rarefaction tail from the filtered spikes via background-subtracted center-of-mass localization.
    \item Sample the far-left, left-star, right-star, and far-right plateau states by trimmed medians.
    \item Perform two sampled-pressure refinement passes: from the sampled far-left and far-right states and the sampled star-side pressures, form a local pressure equation $F(p)=f_L(p)+f_R(p)+u_R-u_L$, take one Newton correction, recover the common contact speed and star states, and update the wave geometry.
    \item Reconstruct $u_\sharp$, $p_\sharp$, and $s_\sharp$ piecewise, use simple-wave interpolation in the fan, and recover $\rho_\sharp$ and $e_\sharp$ from the equation of state.
\end{enumerate}
\end{algorithm}

\begin{remark}[Current wave-pattern scope]
\label{rem:wave-pattern-scope}
The present implementation assumes the classical $1$-rarefaction/$2$-contact/$3$-shock ordering used by the three benchmark problems. Extending the same DRV localization idea to the other one-dimensional ideal-gas Riemann configurations is straightforward in principle: one would keep Steps~1--5 unchanged and alter only the branch logic in Step~6 together with the associated reconstruction formulas. Section~\ref{sec:generalized} demonstrates this extension on four additional test problems covering all four wave-pattern configurations.
\end{remark}

The key conceptual point is that the DRV stage is not being asked to determine all wave positions exactly by itself. Its role is to provide a physically structured, family-selective, sub-cell initialization of the wave geometry. Plateau sampling and the local star-state closure then complete that geometry.

\subsection{Role of the local star-state closure}
\label{sec:local-riemann-closure}

The use of a local star-state closure in Step~6 deserves a brief explanation. The algorithm is not reusing the known exact self-similar solution of the original shock-tube benchmark. The geometric information comes first from the DRV stage: the filtered spikes identify the $1$-wave, the contact, and the $3$-wave and thereby determine the windows from which the plateau states are sampled. In the main implementation, the sampled far-left and far-right states define the pressure equation
\[
F(p)=f_L(p)+f_R(p)+u_R-u_L,
\]
while the sampled left- and right-star pressures are used to seed the initial guess
\[
p_*^{(0)}=\tfrac12\bigl(p_{*,L}^{\mathrm{samp}}+p_{*,R}^{\mathrm{samp}}\bigr).
\]
The current code then takes one Newton correction
\[
p_*^{(1)}=p_*^{(0)}-\tfrac{F(p_*^{(0)})}{F'(p_*^{(0)})},
\]
and recovers the common contact speed from
\[
u_* = \tfrac12\bigl(u_L+u_R+f_R(p_*^{(1)})-f_L(p_*^{(1)})\bigr).
\]
The left and right star densities, as well as the associated wave speeds, are then computed from the usual shock/rarefaction branches determined by the sign of $p_*^{(1)}-p_L$ and $p_*^{(1)}-p_R$.

This closure is available precisely because the DRV stage provides a canonical family-separated $1$--$2$--$3$ decomposition. A jump-like reconstruction such as WENO-Z--THINC--BVD can sharpen the state profile, but it does not naturally output separate detectors for the left acoustic wave, the contact, and the right acoustic wave, nor does it provide comparably distinguished plateau windows attached to those families. Without that DRV-guided geometric decomposition, there is no equally canonical final-time pressure equation to solve.

The one-step closure is also not the only option. For the present benchmarks, one Newton correction is already enough to recover the sharp local geometry reported in Section~\ref{sec:results}. For more general flows, however, the same pressure equation can simply be iterated further. Once the shock/rarefaction branches on the two sides are allowed to switch according to the sign of $p_*-p_k$, the same local closure can accommodate more general wave patterns. In the limit of full convergence, one recovers the exact local ideal-gas Riemann solve discussed separately in Appendix~\ref{app:exact-closure}. This is why the exact local solve should be viewed as an optional sharpening module, not as the source of the DRV geometry itself.

The remainder of this section records the precise formulas and discretization choices used in the accompanying Python implementation, making Algorithm~\ref{alg:recon} directly reproducible. The baseline conservative solver (WENO--5, HLLC, SSP--RK3) is standard and is recorded separately in Appendix~\ref{app:implementation}. The code operates on a uniform mesh of $N$ cells on a domain $[a,b]$, with
\[
\Delta x = \tfrac{b-a}{N},
\qquad
x_j=a+\Bigl(j+\tfrac12\Bigr)\Delta x,
\qquad j=0,\dots,N-1.
\]
In all benchmark families reported in the paper, the domain is $[a,b]=[-\tfrac12,\tfrac12]$, except for the Collision problem which uses $[0,1]$.

\subsection{Discrete algebraic DRV surrogates}

At the final time, the code computes centered differences
\[
D_0 a_j = \tfrac{a_{j+1}-a_{j-1}}{2\Delta x},
\]
using edge padding. The algebraic DRV surrogates are then
\begin{subequations}
\label{eq:alg-drvs-discrete}
\begin{align}
\mrs_j &= D_0 s_j,\\
\mrw_j &= D_0 u_j + \tfrac{\alpha\sigma_j}{\gamma p_j}D_0 p_j,\\
\mrz_j &= D_0 u_j - \tfrac{\alpha\sigma_j}{\gamma p_j}D_0 p_j,
\end{align}
\end{subequations}
where
\[
\sigma_j=\tfrac{1}{\alpha}\sqrt{\tfrac{\gamma p_j}{\rho_j}}.
\]
By Lemma~\ref{lem:alg-identity}, \eqref{eq:alg-drvs-discrete} is exactly the smooth DRV system with derivatives replaced by centered differences.

\subsection{Wave sensor and adaptive Gaussian filtering}

The code smooths the raw algebraic surrogates before spike detection. The smoothing is adaptive. Define
\[
\log\rho_j=\log(\max(\rho_j,\rho_{\min})),
\qquad
\log p_j=\log(\max(p_j,p_{\min})),
\]
and the entropy gradient surrogate
\[
s_{x,j}=\tfrac{D_0p_j}{\max(p_j,p_{\min})}-\gamma\tfrac{D_0\rho_j}{\max(\rho_j,\rho_{\min})}.
\]
The four raw sensor components are
\begin{align*}
s^{(1)}_j &= \operatorname{RN}\bigl(\Delta x\,D_0(\log\rho)_j\bigr),\\
s^{(2)}_j &= \operatorname{RN}\bigl(\Delta x\,D_0(\log p)_j\bigr),\\
s^{(3)}_j &= \operatorname{RN}\!\left(\tfrac{\Delta x\,D_0u_j}{c_j+|u_j|+1}\right),\\
s^{(4)}_j &= \operatorname{RN}\bigl(\Delta x\,s_{x,j}\bigr),
\end{align*}
where $\operatorname{RN}$ denotes the robust normalization used in the code:
\[
\operatorname{RN}(v)_j=
\min\!\left(\tfrac{|v_j|}{Q_{0.90}(|v|)+10^{-14}},1\right),
\]
with fallback to the maximum norm if the $0.90$ quantile degenerates. The sensor is
\[
\mathrm{sensor}_j = G_{1.0}\!\left(\max\{s^{(1)}_j,s^{(2)}_j,s^{(3)}_j,s^{(4)}_j\}\right),
\]
clipped to $[0,1]$, where $G_{1.0}$ is a Gaussian filter of width one cell, truncated at four standard deviations and applied with reflection at the boundary.

For any field $a_j$, the adaptive filtered field is a convex combination of three Gaussian smoothings:
\[
a^{\mathrm{small}}=G_{\sigma_{\min}}a,
\qquad
a^{\mathrm{mid}}=G_{\sigma_{\mathrm{mid}}}a,
\qquad
a^{\mathrm{large}}=G_{\sigma_{\max}}a,
\]
with
\[
\sigma_{\min}=1.25,
\qquad
\sigma_{\mathrm{mid}}=3.0,
\qquad
\sigma_{\max}=5.25
\quad\text{cells}.
\]
The local blending weights are
\[
w_{\mathrm{small}}=s^2,
\qquad
w_{\mathrm{mid}}=2s(1-s),
\qquad
w_{\mathrm{large}}=(1-s)^2,
\]
with $s=\mathrm{sensor}_j$, and the filtered field is
\[
\widetilde a_j=
\tfrac{w_{\mathrm{small}}a^{\mathrm{small}}_j+w_{\mathrm{mid}}a^{\mathrm{mid}}_j+w_{\mathrm{large}}a^{\mathrm{large}}_j}{w_{\mathrm{small}}+w_{\mathrm{mid}}+w_{\mathrm{large}}+10^{-14}}.
\]
When adaptive filtering is disabled, the code instead uses a single Gaussian width $\sigma_{\mathrm{fixed}}=4.0$ cells. In all cases the filter uses reflection at the boundaries and truncation at four standard deviations.

\subsection{Background-subtracted spike localization}

For a filtered spike $g_j$ on an index interval $I=[j_L,j_R]$, define the affine background $b_j$ by linear interpolation between $g_{j_L}$ and $g_{j_R}$, and define the discrete center of mass
\begin{equation}
\widehat X(g;I)
=
\tfrac{\sum_{j=j_L}^{j_R} x_j(g_j-b_j)\,\Delta x}{\sum_{j=j_L}^{j_R}(g_j-b_j)\,\Delta x}.
\label{eq:cm-discrete}
\end{equation}
This is the location formula used everywhere in the code.

\begin{lemma}[Filtered-spike localization]
\label{lem:filtered-spike}
Assume that on a search window $I$ one has a decomposition
\[
g_j = a\,\psi_\eta(x_j-X_*) + r_j,
\]
where $\psi_\eta(x)=\eta^{-1}\psi(x/\eta)$, $\psi$ is smooth and exponentially localized, $\int_{\R}\psi(x)\,dx=m_0\neq 0$, $\int_{\R}x\psi(x)\,dx=0$, and $r_j$ is a perturbation. Assume also that the affine background subtraction removes the far-field linear trend. Then
\[
\widehat X(g;I)-X_*
=
O\!\left(
\Delta x^r
+
\tfrac{\|r\|_{\ell^1(I)}}{|a m_0|}
+
e^{-c\,\mathrm{dist}(X_*,\partial I)/\eta}
\right),
\]
where $r$ is the order of the quadrature implicitly used in \eqref{eq:cm-discrete}.
\end{lemma}

\begin{proof}
Writing
\[
M_I=\sum_{j\in I}(g_j-b_j)\Delta x,
\qquad
N_I=\sum_{j\in I}x_j(g_j-b_j)\Delta x,
\]
one has
\[
M_I = a\sum_{j\in I}\psi_\eta(x_j-X_*)\Delta x + O(\|r\|_{\ell^1(I)}),
\]
and the discrete sum approximates $\int_I\psi_\eta(x-X_*)\,dx$ with the stated quadrature and truncation errors. Likewise,
\[
N_I = a\sum_{j\in I}x_j\psi_\eta(x_j-X_*)\Delta x + O(\|r\|_{\ell^1(I)}).
\]
Because the first moment of $\psi$ vanishes, the leading term is $aX_*m_0$, again up to the same quadrature and truncation errors. Taking the quotient $N_I/M_I$ yields the claim.
\end{proof}

Lemma~\ref{lem:filtered-spike} justifies the center-of-mass localization used in Step~4 of Algorithm~\ref{alg:recon}. The code does not assume the search interval $I$ in advance. Instead it begins from a peak index $j_0$, sets a relative amplitude threshold equal to $0.05$ times the local peak magnitude, and expands left and right while the field remains above threshold. To prevent the window from running into a larger neighboring feature, the expansion is stopped if the neighboring point starts growing again after the window already exceeds three cells on that side. If the expansion collapses to a single cell, the code forces a minimum half-width of three cells. This gives the interval $I$, after which \eqref{eq:cm-discrete} is applied.

\subsection{Initial wave detection}

Let $\mrw^s$, $\mrz^s$, and $\mrs^s$ denote the adaptively filtered DRV surrogates. The code extracts the initial geometry in the following order.
\begin{enumerate}[label=\textnormal{(\alph*)},leftmargin=*]
    \item \textsf{Shock.} The shock is the dominant negative spike of $\mrw^s$: the code sets $j_{\mathrm{shock}}=\arg\min \mrw^s_j$ and then applies the spike-localization procedure above.
    \item \textsf{Contact.} The code searches for the dominant spike of $|\mrs^s|$ to the \emph{left} of the shock. Concretely, if $i_L^{\mathrm{shock}}$ is the left endpoint of the detected shock window, then all indices $\ge i_L^{\mathrm{shock}}-5$ are suppressed in the contact search. If the remaining maximum of $|\mrs^s|$ is less than $10^{-4}$, the contact is declared absent; otherwise the signed field $\mrs^s$ is passed to the center-of-mass routine.
    \item \textsf{Rarefaction head and tail.} The code differentiates the filtered 1-family field,
    \[
    D_0\mrz^s = \tfrac{\mrz^s_{j+1}-\mrz^s_{j-1}}{2\Delta x},
    \]
    and searches to the left of the contact (if the contact has been found). If the positive and negative spikes of $D_0\mrz^s$ are both present with magnitudes exceeding $10^{-5}$, their center-of-mass locations define the rarefaction head and tail.
\end{enumerate}

When the two $D_0\mrz^s$ spikes are not cleanly separated, the code falls back to a support detector based on the positive velocity gradient. It filters $\max(D_0u,0)$ by the same adaptive Gaussian operator, computes its maximum amplitude on the region left of the contact, and takes the support where that filtered positive gradient exceeds $0.05$ of its local maximum. The support endpoints supply fallback rarefaction head and tail candidates. The code chooses between the DRV-based and support-based rarefaction geometry by the explicit heuristics implemented in Python: it switches to the support detector if the DRV fan width is less than $8\Delta x$ while the support width exceeds $20\Delta x$, or if the DRV head lies more than $20\Delta x$ to the right of the support head.

Finally, the geometry is clamped to the physical domain edges and forced to respect the order
\[
X_{rh}\le X_{rt}\le X_c\le X_s,
\]
with at least one cell-width separation when two neighboring waves are both present.

\subsection{Plateau sampling and pressure-wave-function closure}
\label{sec:plateau-closure}

The initial geometry supplies four plateau intervals: the far-left state, the left-star plateau, the right-star plateau, and the far-right state. The code samples each plateau by medians of $\rho$, $u$, and $p$ over a trimmed interior interval.
\begin{itemize}[leftmargin=*]
    \item[$\ast$] If an interval contains more than $2\times 3+1$ cells, the first three and last three cells are discarded before taking medians.
    \item[$\ast$] For the far-left and far-right default states, the medians are taken over the first or last ten cells, respectively (capped between three cells and one-eighth of the domain length).
    \item[$\ast$] The interval endpoints are shifted by $\pm 2\Delta x$ away from the tracked waves before medians are taken, exactly as in the code.
\end{itemize}
If one of the intervals degenerates, the corresponding plateau falls back to the nearest available default state.

Given sampled left and right states $(\rho_L,u_L,p_L)$ and $(\rho_R,u_R,p_R)$ together with sampled star-side pressures $p_{*,L}^{\mathrm{samp}}$ and $p_{*,R}^{\mathrm{samp}}$, the main implementation forms the classical pressure-wave-function equation
\[
F(p)=f_L(p)+f_R(p)+u_R-u_L,
\]
where for $k\in\{L,R\}$, 
\[
f_k(p)=
\begin{cases}
\tfrac{2c_k}{\gamma-1}\left[\left(\tfrac{p}{p_k}\right)^{\tfrac{\gamma-1}{2\gamma}}-1\right], & p\le p_k \quad\text{(rarefaction)},\\[2ex]
(p-p_k)\sqrt{\tfrac{2}{(\gamma+1)\rho_k\left(p+\tfrac{\gamma-1}{\gamma+1}p_k\right)}}, & p>p_k \quad\text{(shock)}.
\end{cases}
\]
Their derivatives are
\[
f_k'(p)=
\begin{cases}
\tfrac{c_k}{\gamma p_k}\left(\tfrac{p}{p_k}\right)^{-\tfrac{\gamma+1}{2\gamma}}, & p\le p_k,\\[2ex]
\sqrt{\tfrac{2}{(\gamma+1)\rho_k\left(p+\tfrac{\gamma-1}{\gamma+1}p_k\right)}}\left(1-\tfrac{p-p_k}{2\left(p+\tfrac{\gamma-1}{\gamma+1}p_k\right)}\right), & p>p_k.
\end{cases}
\]
The main-body closure uses the sampled-pressure seed
\[
p_*^{(0)}=\tfrac12\bigl(p_{*,L}^{\mathrm{samp}}+p_{*,R}^{\mathrm{samp}}\bigr)
\]
and takes a single Newton correction
\[
p_*^{(1)}=p_*^{(0)}-\tfrac{F(p_*^{(0)})}{F'(p_*^{(0)})}.
\]
The code then clips $p_*^{(1)}$ to the pattern-consistent interval
\[
p_R(1+10^{-6})\le p_*\le p_L(1-10^{-6}),
\]
whenever these bounds are compatible, and otherwise only enforces positivity. The common contact speed is recovered from
\[
u_* = \tfrac12\bigl(u_L+u_R+f_R(p_*)-f_L(p_*)\bigr).
\]

The star densities are then computed from the usual ideal-gas branches
\[
\rho_{*k}=
\begin{cases}
\rho_k\left(\tfrac{p_*}{p_k}\right)^{1/\gamma}, & p_*\le p_k,\\[2ex]
\rho_k\,\tfrac{\tfrac{p_*}{p_k}+\tfrac{\gamma-1}{\gamma+1}}{\tfrac{\gamma-1}{\gamma+1}\tfrac{p_*}{p_k}+1}, & p_*>p_k,
\end{cases}
\]
and the corresponding wave speeds are obtained from the same branch decision. Thus
\[
X_{rh}=x_*+s_{L,h}t,\qquad
X_{rt}=x_*+s_{L,t}t,\qquad
X_c=x_*+u_*t,\qquad
X_s=x_*+s_{R,t}t,
\]
where for a rarefaction one uses the head and tail characteristic speeds and for a shock one uses the standard shock speed. In the present benchmarks this reproduces the expected $1$-rarefaction/$2$-contact/$3$-shock structure. The implementation repeats the refinement pass twice. Appendix~\ref{app:exact-closure} records the fully converged exact local closure as an optional variant.

\subsection{Primitive reconstruction}

Once the final geometry and plateau states have been determined, the code reconstructs $u_\sharp$, $p_\sharp$, and $s_\sharp$ piecewise. Outside the rarefaction fan, the reconstruction is piecewise constant on the left plateau, the left-star plateau, the right-star plateau, and the right plateau. On the rarefaction interval $[X_{rh},X_{rt}]$, the code uses the self-similar simple-wave structure.

Let $\theta=(x-X_{rh})/(X_{rt}-X_{rh})$. Inside the fan,
\[
u_\sharp(x)=u_L+\theta(u_{*L}-u_L),
\qquad
c_\sharp(x)=c_L+\theta(c_{*L}-c_L),
\qquad
s_\sharp(x)=s_L,
\]
and therefore
\[
p_\sharp(x)=p_L\left(\tfrac{c_\sharp(x)}{c_L}\right)^{\tfrac{2\gamma}{\gamma-1}},
\qquad
\rho_\sharp(x)=\rho_L\left(\tfrac{c_\sharp(x)}{c_L}\right)^{\tfrac{2}{\gamma-1}}.
\]
In the actual code, $u_\sharp$, $p_\sharp$, and $s_\sharp$ are first reconstructed, and then
\[
s_\sharp \leftarrow \min\{80,\max(-80,s_\sharp)\},
\qquad
p_\sharp \leftarrow \max(p_\sharp,p_{\min}),
\]
before the thermodynamic closure
\begin{equation}
\rho_\sharp = \left(\tfrac{p_\sharp}{e^{s_\sharp}}\right)^{1/\gamma},
\qquad
e_\sharp = \tfrac{p_\sharp}{(\gamma-1)\rho_\sharp}.
\label{eq:eos-sharp}
\end{equation}

\begin{table}[t]
\centering
\small
\begin{tabular}{ll}
\toprule
quantity & value used in the code \\
\midrule
$\rho_{\min}$, $p_{\min}$, denominator floor & $10^{-14}$ \\
WENO regularization $\varepsilon_{\mathrm{WENO}}$ & $10^{-6}$ \\
robust-normalization quantile & $0.90$ \\
Gaussian truncate parameter & $4.0$ \\
Gaussian boundary mode & reflect \\
adaptive widths & $\sigma_{\min}=1.25$, $\sigma_{\mathrm{mid}}=3.0$, $\sigma_{\max}=5.25$ cells \\
fixed width (when adaptive mode is off) & $\sigma_{\mathrm{fixed}}=4.0$ cells \\
spike relative threshold & $0.05$ of local peak \\
minimum spike half-width & $3$ cells \\
contact-detection amplitude threshold & $10^{-4}$ \\
$D_0\mrz$ amplitude threshold & $10^{-5}$ \\
plateau trim & $3$ cells \\
edge plateau median width & $10$ cells \\
Newton corrections per refinement pass & $1$ \\
pressure clip factor in Step~6 & $1\pm 10^{-6}$ \\
number of refinement passes & $2$ \\
\bottomrule
\end{tabular}
\caption{Implementation constants needed to reproduce the accompanying Python code.}
\label{tab:impl-constants}
\end{table}

\subsection{Relation to existing reconstruction schemes}
\label{sec:related}

The broad goal of sharpening discontinuities beyond a standard Godunov or WENO baseline is not new. There is a large existing literature, which can be organized into three families: contact-steepening and artificial-compression methods built into the evolution scheme \cite{ColellaWoodward1984,Huynh1995,Yang1990,JiangShu1996}; subcell-resolution methods that estimate within-cell discontinuity positions from cell averages \cite{HartenSubcell1989,ShuOsher1989}; and interface-reconstruction methods, including non-diffusive contact-capturing schemes \cite{DespresLagoutiere2001,Aguillon2014} and THINC/BVD-type jump-like reconstructions that represent a discontinuity partly inside a cell \cite{SunInabaXiao2016}.

The present method belongs to this general neighborhood, but the specific mechanism is different: the localization stage is based on differentiated characteristic variables that approximate the wave-family Radon measures, rather than on steepening the state profile or choosing between polynomial and jump-like state reconstructions. The three detectors are family-separated ($\mrz$ for the $1$-wave, $\mrs$ for the contact, $\mrw$ for the $3$-wave), and the primitive variables are reconstructed only after that geometric information has been extracted and refined by a local pressure-wave-function closure.

The LeBlanc internal-energy defect provides an especially stringent test. Many standard and higher-order schemes still exhibit a visible overshoot there \cite{WangDuZhaoYuan2020}. The nonlinear artificial viscosity scheme, the $C$-method,  that we introduced in  \cite{ReisnerSerencsaShkoller2013,RamaniReisnerShkoller2019} is capable of  suppressing that overshoot but it is not a light final-time postprocessor. The direct comparison in Section~\ref{sec:comparison-results} tests the present mechanism against two strong non-DRV comparators---a MUSCL--THINC--BVD hybrid \cite{vanLeer1979,Toro2009,SunInabaXiao2016} and a WENO-Z--THINC--BVD hybrid \cite{Borges2008,SunInabaXiao2016}---and shows that neither reproduces the combination of near-discontinuous contacts, small contact-window internal-energy error, and suppression of the LeBlanc positive overshoot achieved by the DRV reconstruction.

\section{Reconstruction eliminates the contact-layer internal-energy overshoot}
\label{sec:contact-prop}

A well-known defect of conservative shock-capturing schemes on strong Riemann problems is that the numerically smeared contact layer can produce a spurious overshoot in the specific internal energy $e=p/((\gamma-1)\rho)$. The overshoot arises because the baseline scheme transitions between the left-star and right-star states through a numerically diffused layer in which $u$, $p$, $\rho$, and $s$ all vary simultaneously. Since $e$ depends nonlinearly on $\rho$ and $p$, the transition can create intermediate values of $e$ that exceed the correct star-state values on both sides. This defect is especially visible in the severe-expansion and LeBlanc tests.

The DRV reconstruction avoids this mechanism by design. Rather than smoothly transitioning through the smeared contact, it reconstructs $u_\sharp$ and $p_\sharp$ as flat on both sides of the tracked contact position $X_c$, with the entropy $s_\sharp$ carrying a single jump. The following proposition makes this precise.

\begin{proposition}[No additional contact-layer spike]
\label{prop:contact-spike}
Assume that, in the closure stage of Algorithm~\ref{alg:recon}, the reconstructed fields near the tracked contact $X_c$ satisfy the following:
\begin{enumerate}[label=\textnormal{(\roman*)},leftmargin=*]
    \item $u_\sharp$ and $p_\sharp$ are constant on both sides of $X_c$ with the same left and right star values $u_*$ and $p_*$;
    \item $s_\sharp$ is piecewise constant with a single jump at $X_c$;
    \item $\rho_\sharp$ and $e_\sharp$ are recovered from \eqref{eq:eos-sharp}.
\end{enumerate}
Then $e_\sharp$ contains no additional contact-centered overshoot. Only the physical jump induced by the entropy jump remains.
\end{proposition}

\begin{proof}
By assumption, $u_\sharp$ and $p_\sharp$ are flat across the contact, so there is no transition layer in those fields from which an artificial thermodynamic spike could be generated. Since $s_\sharp$ is piecewise constant with a single jump, the density recovered from
\[
\rho_\sharp=\left(\tfrac{p_*}{e^{s_\sharp}}\right)^{1/\gamma}
\]
is also piecewise constant with only the physical contact jump. Consequently
\[
e_\sharp=\tfrac{p_*}{(\gamma-1)\rho_\sharp}
\]
is piecewise constant as well, again with only the physical contact jump. There is therefore no extra contact-centered overshoot in $e_\sharp$.
\end{proof}

The proposition is intentionally modest: it does not say that the baseline solution is free of wall heating (it is not), nor does it require the closure to be exact. It says only that once the reconstruction replaces the smeared contact layer by a contact-flat pair $(u_\sharp,p_\sharp)$ together with a single entropy jump, the reconstructed internal energy \emph{cannot} develop a new contact-centered spike. This condition is satisfied whenever the closure in Step~6 returns a common star pair $(u_*,p_*)$---whether that closure is the one-step Newton correction of Section \ref{sec:strategy} or the fully converged exact Riemann solve of Appendix~\ref{app:exact-closure}. The severe-expansion and LeBlanc results in Section~\ref{sec:results} confirm that this structural guarantee is realized in practice: the reconstructed internal-energy profiles show the physical entropy-induced jump but no spurious overshoot.

\begin{remark}[Conceptual level of the direct comparisons]
\label{rem:comparison-levels}
The direct comparison in Section~\ref{sec:comparison-results} juxtaposes methods operating at different conceptual levels: MUSCL--THINC--BVD and WENO-Z--THINC--BVD are conservative evolution-time reconstructors, whereas the present DRV procedure is a final-time geometric reconstruction. The comparison should be read as evidence that DRV geometry plus a very cheap local closure already outperforms the two independently coded jump-like competitors on the present tests.
\end{remark}

\section{Numerical results and wave-location errors}
\label{sec:results}

All results in this section were generated with the accompanying Python implementation, whose detailed formulas are recorded in Section~\ref{sec:strategy}, using HLLC for the baseline solve, adaptive filtering, and two sampled-pressure refinement passes with one Newton correction per pass. All reported runs use the common spatial domain $[a,b]=[-\tfrac12,\tfrac12]$. The benchmark resolutions are $N=600$ for Sod and $N=1000$ for the severe-expansion and LeBlanc tests; the direct comparison study in Section~\ref{sec:comparison-results} is run at the common resolution $N=600$. Appendix~\ref{app:exact-closure} records the optional exact local Riemann closure and compares it against the one-step closure used here. The three benchmark parameter sets are:
\begin{center}
\begin{tabular}{lcccccc}
\toprule
benchmark & $N$ & $t_{\mathrm{final}}$ & $\gamma$ & $(\rho_L,u_L,p_L)$ & $(\rho_R,u_R,p_R)$ & CFL \\
\midrule
Sod & $600$ & $0.15$ & $1.4$ & $(1,0,1)$ & $(0.125,0,0.1)$ & $0.40$ \\
Severe expansion & $1000$ & $0.12$ & $1.4$ & $(1,0,1)$ & $(10^{-4},0,10^{-4})$ & $0.35$ \\
LeBlanc & $1000$ & $0.50$ & $5/3$ & $(1,0,\tfrac23\cdot 10^{-1})$ & $(10^{-3},0,\tfrac23\cdot 10^{-10})$ & $0.15$ \\
\bottomrule
\end{tabular}
\end{center}
The interface locations are $x_*=0$ for Sod, $x_*=-0.2$ for the severe expansion test, and $x_*=-\tfrac13$ for LeBlanc.

For each benchmark, the exact reference locations are taken from the standard exact self-similar ideal-gas Riemann solution. We compare the initial DRV-based geometry to the refined geometry through the absolute errors
\[
E_{\bullet}^{\mathrm{init}} = |X_{\bullet}^{\mathrm{init}}-X_{\bullet}^{\mathrm{ex}}|,
\qquad
E_{\bullet}^{\mathrm{ref}} = |X_{\bullet}^{\mathrm{ref}}-X_{\bullet}^{\mathrm{ex}}|,
\]
for $\bullet\in\{rh,rt,c,s\}$.

\subsection{Exact wave locations and errors}

Table~\ref{tab:wave-errors} summarizes the wave-location errors. The main message is that the two stages of the algorithm play different roles. The DRV stage already orders the waves correctly and places them at sub-cell accuracy. The one-step sampled-pressure closure is what collapses the remaining geometric error once the plateau states are sampled correctly. In Sod, the refined geometry agrees with the exact Riemann solution to within roundoff ($O(10^{-14})$). In the severe-expansion test, the fan-tail, contact, and shock errors drop from $2.54\times 10^{-2}$, $1.92\times 10^{-2}$, and $3.24\times 10^{-2}$ to $1.90\times 10^{-4}$, $2.07\times 10^{-4}$, and $5.59\times 10^{-4}$, respectively. In LeBlanc, the fan tail, contact, and shock errors all drop to the $10^{-4}$ level, while the rarefaction head remains the hardest feature because it lies on the left boundary and is therefore controlled mainly by boundary-side sampling rather than by an interior spike.

\begin{table}[t]
\centering
\small
\begin{tabular}{llccc}
\toprule
benchmark & wave & exact position & $E^{\mathrm{init}}$ & $E^{\mathrm{ref}}$ \\
\midrule
Sod & fan head & $-0.177482$ & $2.315\times 10^{-3}$ & $0$ \\
Sod & fan tail & $-0.010541$ & $2.983\times 10^{-3}$ & $2.597\times 10^{-15}$ \\
Sod & contact & $0.139118$ & $1.245\times 10^{-4}$ & $1.132\times 10^{-14}$ \\
Sod & shock & $0.262823$ & $1.063\times 10^{-3}$ & $8.110\times 10^{-14}$ \\
\midrule
Severe expansion & fan head & $-0.341986$ & $1.051\times 10^{-2}$ & $0$ \\
Severe expansion & fan tail & $0.167113$ & $2.539\times 10^{-2}$ & $1.897\times 10^{-4}$ \\
Severe expansion & contact & $0.224249$ & $1.924\times 10^{-2}$ & $2.066\times 10^{-4}$ \\
Severe expansion & shock & $0.346020$ & $3.239\times 10^{-2}$ & $5.592\times 10^{-4}$ \\
\midrule
LeBlanc & fan head & $-0.500000$ & $5.000\times 10^{-4}$ & $2.930\times 10^{-3}$ \\
LeBlanc & fan tail & $-0.085441$ & $1.194\times 10^{-2}$ & $1.689\times 10^{-4}$ \\
LeBlanc & contact & $-0.022414$ & $1.061\times 10^{-1}$ & $1.842\times 10^{-4}$ \\
LeBlanc & shock & $0.081226$ & $1.003\times 10^{-2}$ & $2.457\times 10^{-4}$ \\
\bottomrule
\end{tabular}
\caption{Exact wave locations and absolute errors for the implemented algorithm. The Sod test uses $N=600$ cells, while the severe-expansion and LeBlanc tests use $N=1000$ cells. The initial geometry comes from filtered DRV spikes; the refined geometry is obtained after two sampled-pressure refinement passes with one Newton correction per pass.}
\label{tab:wave-errors}
\end{table}

\subsection{Conservation, grid dependence, and parameter sensitivity}

Because the reconstructed fields are not constrained to reproduce the baseline cell averages exactly, the DRV reconstruction is not a discretely conservative replacement for the evolved finite-volume state; the baseline conservative state is always retained. Table~\ref{tab:conservation-defects} reports the relative mass, momentum, and energy defects of the reconstructed profiles with respect to that baseline. The defects are small: below $2\times 10^{-4}$ for Sod, at the $10^{-4}$--$10^{-3}$ level for the severe-expansion problem, and below $5\times 10^{-3}$ for LeBlanc. To put these numbers in context, the baseline solution's own density $L^1$ error is $8.6\times 10^{-4}$ for Sod, its contact-window internal-energy error is $2.9\times 10^{-1}$ for the severe expansion, and its LeBlanc contact-layer overshoot is $3.5\times 10^{-2}$. The reconstruction reduces all of these errors by one to ten orders of magnitude (see Tables~\ref{tab:wave-errors} and \ref{tab:comparison-bvd}), so the conservation defects are far smaller than the accuracy gains they enable.
\begin{table}[t]
\centering
\small
\begin{tabular}{lccc}
\toprule
benchmark & mass defect & momentum defect & energy defect \\
\midrule
Sod & $-9.51\times 10^{-5}$ & $5.51\times 10^{-5}$ & $1.94\times 10^{-4}$ \\
Severe expansion & $-2.72\times 10^{-4}$ & $-1.26\times 10^{-3}$ & $-4.40\times 10^{-4}$ \\
LeBlanc & $-3.65\times 10^{-3}$ & $-2.25\times 10^{-3}$ & $-4.83\times 10^{-3}$ \\
\bottomrule
\end{tabular}
\caption{Relative conservation defects of the final reconstructed profiles with respect to the baseline finite-volume state at the benchmark resolutions.}
\label{tab:conservation-defects}
\end{table}

To assess grid dependence, the DRV reconstruction was rerun at $N=200,400,800,1600$ on the same domain. Figure~\ref{fig:grid-convergence} shows two trends. First, the one-step closure remains very accurate across the tested grids. For Sod, the refined wave-location error is at roundoff on every tested grid. For LeBlanc, the visible residual is again the boundary-controlled rarefaction head, for which
\[
E_{\max}^{\mathrm{ref}}:=\max\{E_{rh}^{\mathrm{ref}},E_{rt}^{\mathrm{ref}},E_c^{\mathrm{ref}},E_s^{\mathrm{ref}}\}
\]
decreases from $1.11\times 10^{-2}$ at $N=200$ to $2.01\times 10^{-3}$ at $N=1600$. For the severe-expansion problem, the interior wave errors remain small---between $3.10\times 10^{-5}$ and $1.13\times 10^{-3}$ in this study---but are no longer driven to roundoff at intermediate resolutions, which is one of the main differences between the one-step closure and the exact closure recorded in Appendix~\ref{app:exact-closure}. Second, the density contact width scales essentially like one cell, with $W_\rho\approx \Delta x$ on all three tests. In the same grid study, the severe-expansion contact-window internal-energy error stays between $1.0\times 10^{-4}$ and $2.3\times 10^{-3}$, while the LeBlanc positive-overshoot metric $O_e^+$ remains nonpositive on all tested grids.

A light robustness check was also carried out on the severe-expansion and LeBlanc problems. Changing the spike relative threshold from $0.05$ to $0.10$ and scaling all three adaptive Gaussian widths by $\pm 15\%$ produced no visible change in the reported reconstruction metrics. In the implemented runs, the contact and shock errors as well as the contact-window internal-energy metrics remained identical to the displayed digits. This is not a full parameter study, but it suggests that the implemented detector is not finely tuned to a single narrow parameter set.

\begin{figure}[t]
\centering
\includegraphics[width=0.98\textwidth]{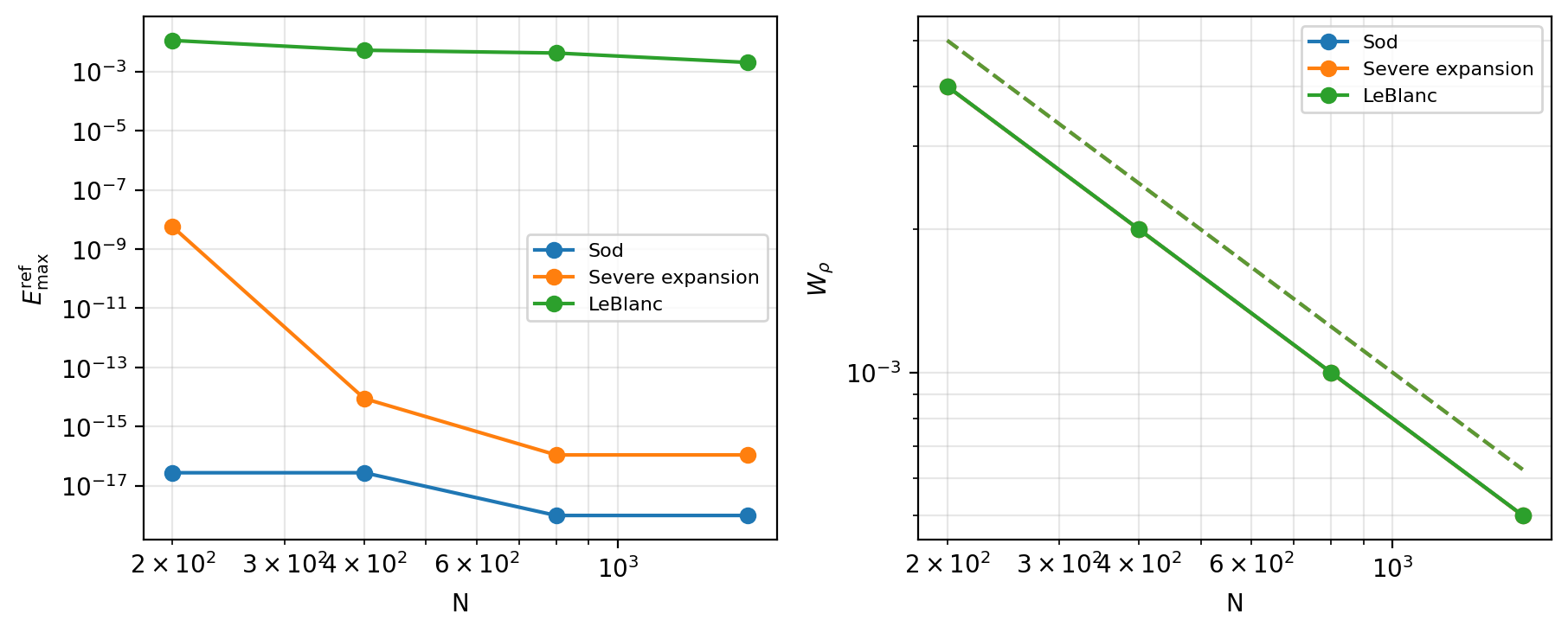}
\caption{Grid-dependence study for the DRV reconstruction on $N=200,400,800,1600$. Left: the maximum refined wave-location error $E_{\max}^{\mathrm{ref}}$. Sod stays at roundoff, the dominant visible residual is the boundary-controlled LeBlanc rarefaction head, and the severe-expansion errors remain small but nonzero for the one-step closure. Right: the $10$--$90$ density contact width $W_\rho$. The dashed reference line is $\Delta x$, showing that the reconstructed contact width remains essentially one cell across the tested resolutions.}
\label{fig:grid-convergence}
\end{figure}

\subsection{Sod shock tube}

For Sod, the baseline HLLC/WENO computation is already visually good. This makes the test useful for isolating the geometric effect of the DRV localization. The initial DRV-based positions are already accurate at roughly the $10^{-3}$ level, and after two sampled-pressure refinement passes the rarefaction head, rarefaction tail, contact, and shock all coincide with the exact Riemann locations to plotting accuracy; see Table~\ref{tab:wave-errors}. Figure~\ref{fig:sod-density} shows that the reconstructed density sharpens the contact and shock while keeping the correct rarefaction geometry.

\begin{figure}[t]
\centering
\includegraphics[width=0.88\textwidth]{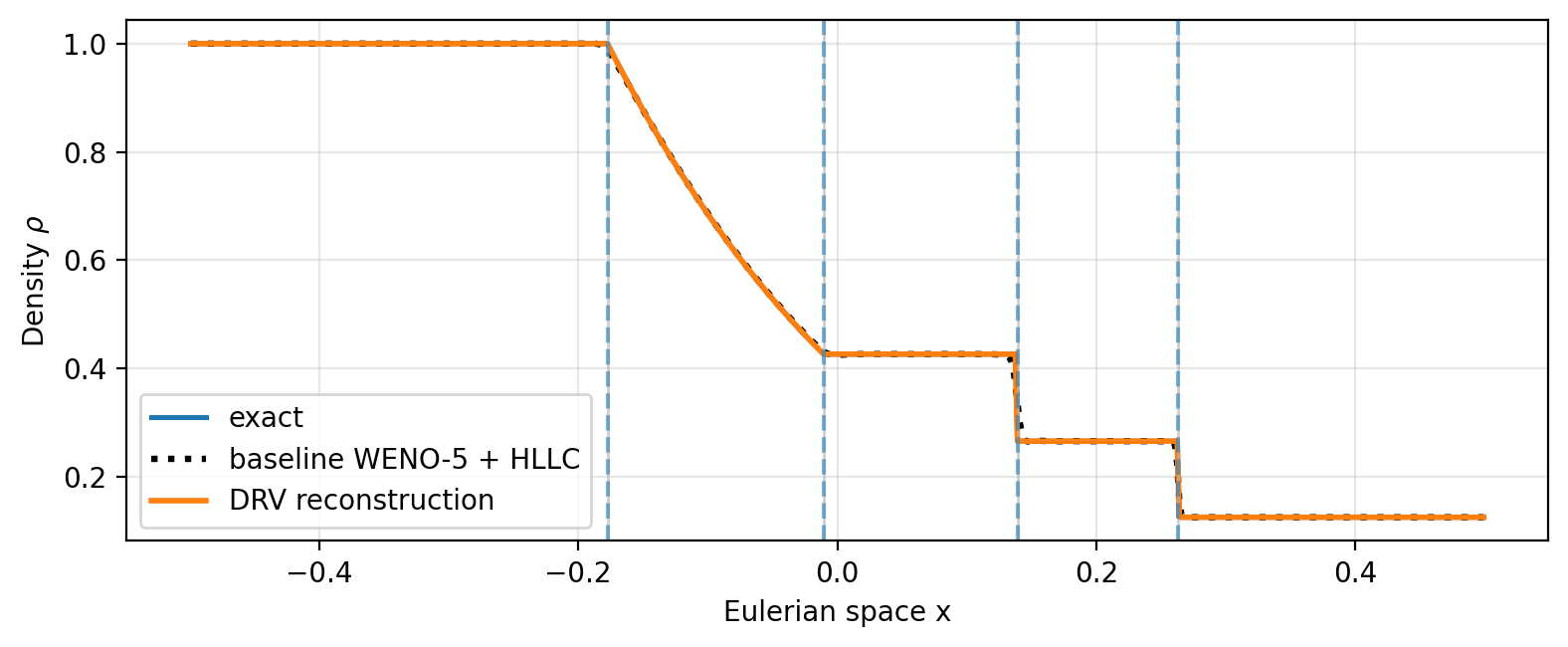}
\caption{Sod shock tube on $N=600$ cells at $t=0.15$. Dashed: exact density. Dotted black: baseline WENO--5+HLLC density. Solid: reconstructed density. The one-step closure aligns the refined wave locations with the exact self-similar Riemann solution to plotting accuracy.}
\label{fig:sod-density}
\end{figure}

\subsection{Severe expansion shock tube}

The severe expansion problem is much more revealing; it is a near-vacuum expansion test in the Tang--Liu/Rider spirit \cite{TangLiu2006,Rider2018}. Here the baseline solution exhibits a strong nonphysical contact-layer defect in the internal energy. The DRV stage still orders the waves correctly, but the raw geometric errors are no longer tiny: the initial fan-head, fan-tail, contact, and shock errors are approximately $1.05\times 10^{-2}$, $2.54\times 10^{-2}$, $1.92\times 10^{-2}$, and $3.24\times 10^{-2}$, respectively. After two one-step refinement passes, those errors drop to $0$, $1.90\times 10^{-4}$, $2.07\times 10^{-4}$, and $5.59\times 10^{-4}$. Figure~\ref{fig:severe-expansion-e} shows the internal energy. The baseline curve contains the well-known contact-layer distortion, whereas the reconstructed profile tracks the exact solution and suppresses the extra contact-centered defect. At the benchmark resolution, the contact-window internal-energy error is $6.00\times 10^{-4}$ and the positive-overshoot metric is nonpositive.

This is the benchmark for which Proposition~\ref{prop:contact-spike} is most visible computationally. The one-step closure already returns a common contact-side pair $(u_*,p_*)$, so the reconstruction uses equal star values of $u$ and $p$ across the contact while retaining the physical entropy jump. The resulting internal-energy graph contains the physical discontinuity but not the spurious contact-centered contamination present in the baseline solution.

\begin{figure}[t]
\centering
\includegraphics[width=0.88\textwidth]{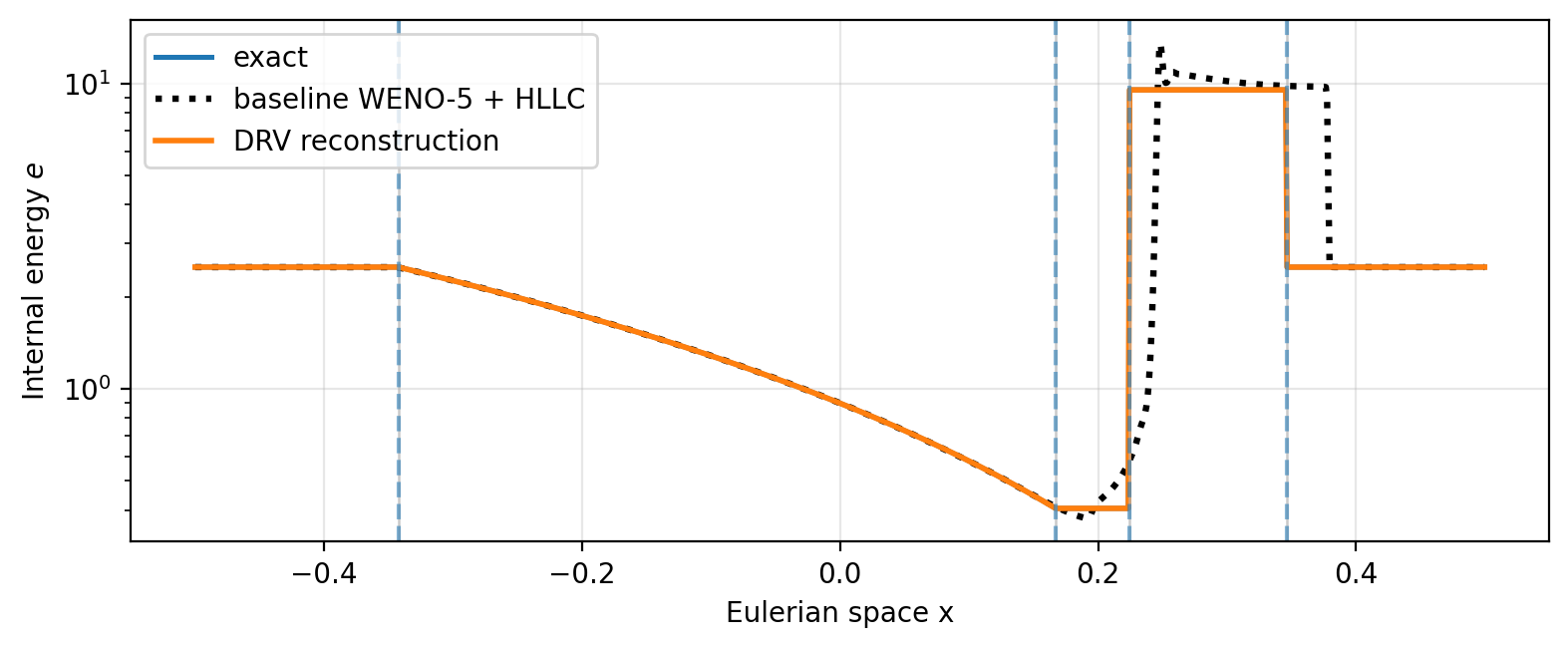}
\caption{Severe expansion shock tube on $N=1000$ cells at $t=0.12$ (logarithmic vertical scale). Dashed: exact internal energy. Dotted black: baseline WENO--5+HLLC internal energy. Solid: reconstructed internal energy. The one-step closure removes the contact-layer defect and aligns the wave locations to within a few $10^{-4}$ to $10^{-3}$.}
\label{fig:severe-expansion-e}
\end{figure}

\subsection{LeBlanc shock tube}

LeBlanc is the hardest test  presented here because of the extreme scale separation between the two sides of the initial data and because the rarefaction head is aligned with the left boundary. The initial contact location is badly displaced: $E_c^{\mathrm{init}}=1.061\times 10^{-1}$. After refinement, that error drops to $1.842\times 10^{-4}$, the fan-tail error drops to $1.689\times 10^{-4}$, and the shock error drops to $2.457\times 10^{-4}$. The rarefaction head is the only feature that does not improve here; its refined error is $2.930\times 10^{-3}$, which is consistent with the fact that the head is not being localized by an interior spike but by the boundary-side geometry.

Even in this more difficult regime, the method does what it is supposed to do. The DRV stage provides enough structure to place the local plateaus correctly, and the one-step sampled-pressure closure then sharply corrects the contact and shock locations. Figure~\ref{fig:leblanc-e} shows the internal energy. Relative to the baseline solution, the reconstructed profile has substantially better wave placement, a contact-window internal-energy error of $1.88\times 10^{-4}$, and no positive contact-layer overshoot in the metric $O_e^+$. Appendix~\ref{app:exact-closure} shows that, for LeBlanc, the exact local Riemann closure is numerically indistinguishable from this one-step version to the displayed digits.

\begin{figure}[t]
\centering
\includegraphics[width=0.88\textwidth]{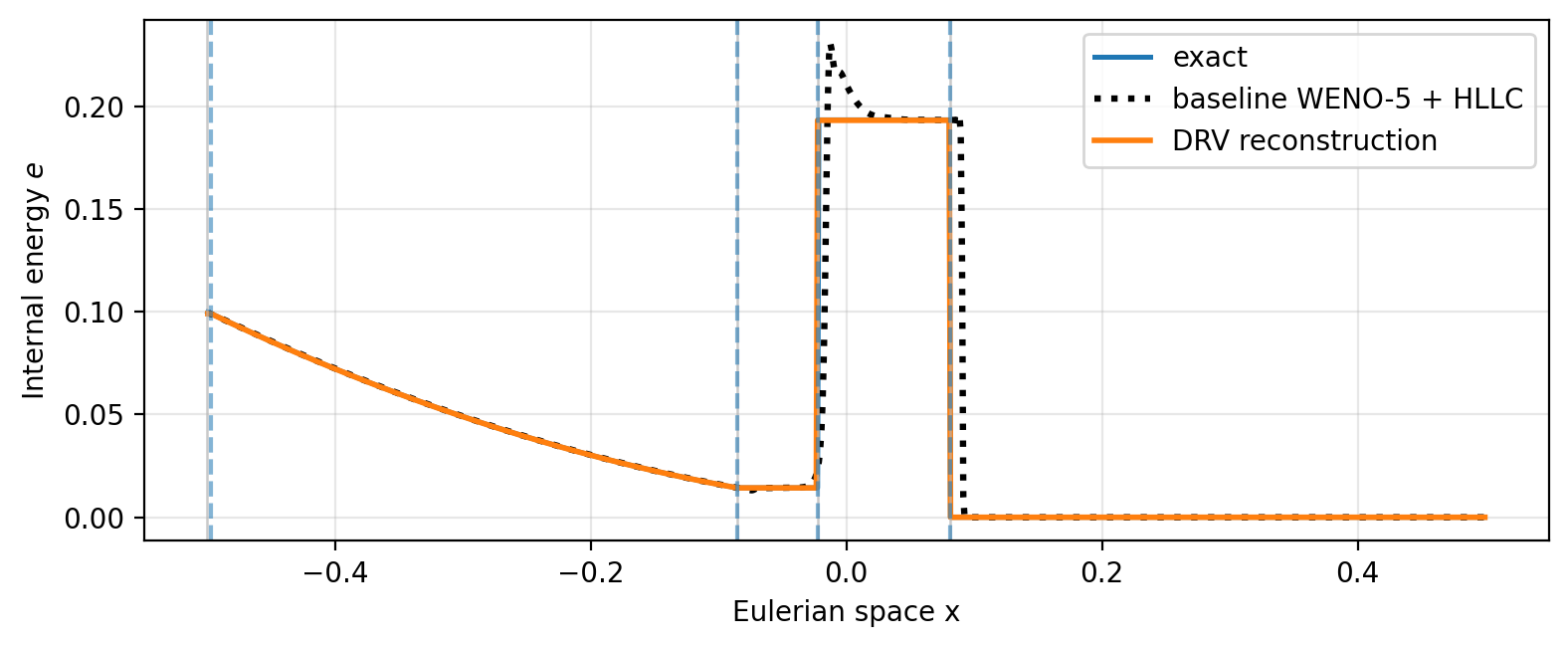}
\caption{LeBlanc shock tube on $N=1000$ cells at $t=0.5$. Dashed: exact internal energy. Dotted black: baseline WENO--5+HLLC internal energy. Solid: reconstructed internal energy. The contact and shock are much better aligned after the one-step closure, and the contact-layer thermodynamic contamination is strongly reduced.}
\label{fig:leblanc-e}
\end{figure}

\subsection{Independent comparison with two jump-like reconstruction schemes}
\label{sec:comparison-results}

The novelty question becomes much sharper when one compares directly against strong non-DRV reconstructions rather than only against broad literature families. To that end, two independent comparison schemes were coded from scratch. The first is a MUSCL--THINC--BVD hybrid: for each primitive variable, a piecewise-linear MUSCL candidate based on the monotonized-central slope \cite{vanLeer1979,Toro2009} is paired with a THINC jump-like candidate, and a BVD rule selects the cellwise interface values so as to reduce the boundary jump \cite{SunInabaXiao2016}. The second replaces the MUSCL candidate by a fifth-order WENO-Z candidate \cite{Borges2008}, producing a WENO-Z--THINC--BVD hybrid. These are not claimed to be verbatim reproductions of every published BVD variant; rather, they are strong independently coded representatives of the jump-like reconstruction philosophy.

Both comparison schemes use the same HLLC flux and SSP--RK3 time integrator as the baseline, the same positivity floors, a fixed THINC steepness parameter $\beta=1.6$, and no benchmark-by-benchmark retuning beyond stable CFL choices. All three reconstructions are then compared at the common resolution $N=600$. To keep the comparison fair, the contact and shock locations reported for the non-DRV schemes are extracted from the computed profiles by the same profile-based detectors: the contact is located by the dominant entropy-gradient spike in the exact-contact window, and the shock is located by the dominant pressure-gradient spike in the exact-shock window. For the DRV method, the results reported in Table~\ref{tab:comparison-bvd} use the same one-step sampled-pressure closure as the rest of the main body. At the plotting scale of Figure~\ref{fig:comparison-panels}, the corresponding exact-closure DRV curves from Appendix~\ref{app:exact-closure} are visually indistinguishable.

Table~\ref{tab:comparison-bvd} and Figure~\ref{fig:comparison-panels} give the result. The comparison has three clear messages. For Sod, WENO-Z--THINC--BVD is respectable, but it is still visibly less sharp than the DRV reconstruction: its density $L^1$ error is $6.63\times 10^{-4}$ and its $10$--$90$ contact width is $5.75\times 10^{-3}$, whereas the DRV reconstruction gives $W_\rho=1.33\times 10^{-3}$. For the severe-expansion problem, both BVD schemes improve substantially on the baseline, but neither comes close to the DRV reconstruction: the contact-window internal-energy error remains about $2.18\times 10^{-1}$ for both jump-like competitors, while the DRV value is $1.68\times 10^{-2}$ with contact and shock location errors $8.23\times 10^{-5}$ and $1.85\times 10^{-3}$. (At the benchmark resolution $N=1000$, the DRV contact-window error drops further to $6.00\times 10^{-4}$; the comparison table uses $N=600$ for all schemes.) For LeBlanc, the contrast is stronger still. MUSCL--THINC--BVD and WENO-Z--THINC--BVD both retain a visible positive contact-layer overshoot, with $O_e^+=4.38\times 10^{-2}$ and $3.69\times 10^{-2}$ respectively, whereas the DRV reconstruction has nonpositive $O_e^+$ in the same metric and simultaneously reduces the contact width to $1.33\times 10^{-3}$.

As noted in Remark~\ref{rem:comparison-levels}, this is not an apples-to-apples comparison of conservative evolution schemes. Its value is practical. Among two strong jump-like schemes chosen precisely because they should have been serious challengers, neither reproduces the combination of sharp contact reconstruction, small contact-window internal-energy error, and elimination of the LeBlanc positive overshoot achieved by the DRV pipeline when that pipeline uses only the one-step local closure.

\begin{table}[t]
\centering
\scriptsize
\setlength{\tabcolsep}{4pt}
\resizebox{\textwidth}{!}{%
\begin{tabular}{lccccc ccc}
\toprule
& \multicolumn{2}{c}{Sod} & \multicolumn{3}{c}{Severe expansion} & \multicolumn{3}{c}{LeBlanc} \\
\cmidrule(lr){2-3}\cmidrule(lr){4-6}\cmidrule(lr){7-9}
scheme & $\|\rho-\rho^{\mathrm{ex}}\|_{L^1}$ & $W_\rho$ & $E_c$ & $E_s$ & $E_{e,\mathrm{ct}}$ & $W_\rho$ & $E_{e,\mathrm{ct}}$ & $O_e^+$ \\
\midrule
baseline WENO--5/HLLC & $8.65\times 10^{-4}$ & $7.62\times 10^{-3}$ & $3.16\times 10^{-2}$ & $4.31\times 10^{-2}$ & $2.90\times 10^{-1}$ & $5.40\times 10^{-2}$ & $2.10\times 10^{-3}$ & $3.50\times 10^{-2}$ \\
MUSCL--THINC--BVD & $1.01\times 10^{-3}$ & $9.49\times 10^{-3}$ & $2.16\times 10^{-2}$ & $2.15\times 10^{-2}$ & $2.18\times 10^{-1}$ & $5.19\times 10^{-2}$ & $2.29\times 10^{-3}$ & $4.38\times 10^{-2}$ \\
WENO-Z--THINC--BVD & $6.63\times 10^{-4}$ & $5.75\times 10^{-3}$ & $2.16\times 10^{-2}$ & $3.31\times 10^{-2}$ & $2.18\times 10^{-1}$ & $1.55\times 10^{-2}$ & $1.57\times 10^{-3}$ & $3.69\times 10^{-2}$ \\
DRV reconstruction & $0$ & $1.33\times 10^{-3}$ & $8.23\times 10^{-5}$ & $1.85\times 10^{-3}$ & $1.68\times 10^{-2}$ & $1.33\times 10^{-3}$ & $3.09\times 10^{-4}$ & none \\
\bottomrule
\end{tabular}}
\caption{Independent comparison at the common resolution $N=600$ cells. For Sod, $\|\rho-\rho^{\mathrm{ex}}\|_{L^1}$ is the density $L^1$ error and $W_\rho$ is the $10$--$90$ density transition width at the contact. For the severe-expansion problem, $E_c$ and $E_s$ are the profile-based contact and shock location errors obtained uniformly from entropy-gradient and pressure-gradient detectors, while $E_{e,\mathrm{ct}}$ is the internal-energy $L^1$ error on a fixed contact window. For LeBlanc, $W_\rho$ and $E_{e,\mathrm{ct}}$ are the analogous contact-width and contact-window internal-energy metrics, and $O_e^+$ is the positive contact-window overshoot relative to the larger exact star-state internal energy (``none'' = no positive overshoot detected). The DRV row uses the one-step sampled-pressure closure of the main body.}
\label{tab:comparison-bvd}
\end{table}

\begin{figure}[p]
\centering
\includegraphics[width=0.88\textwidth]{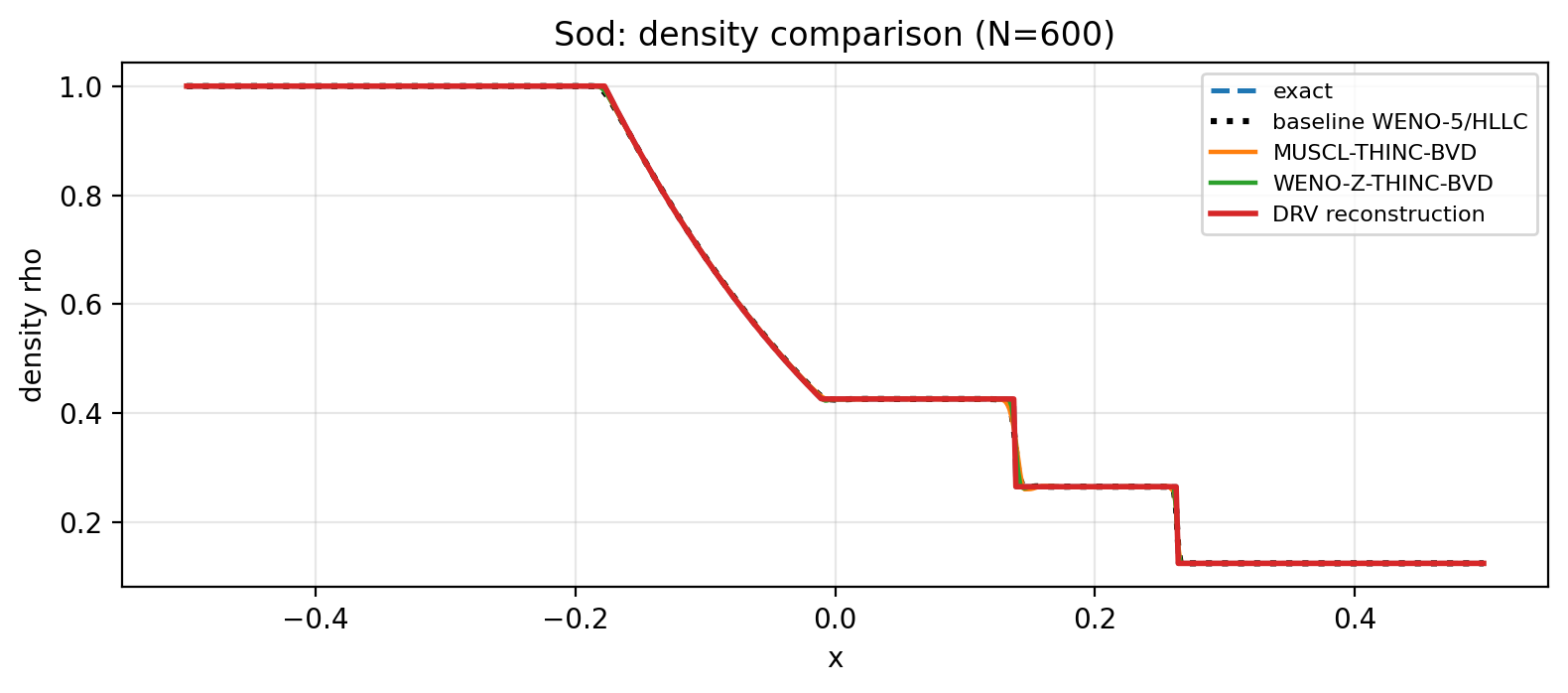}

\medskip
\includegraphics[width=0.88\textwidth]{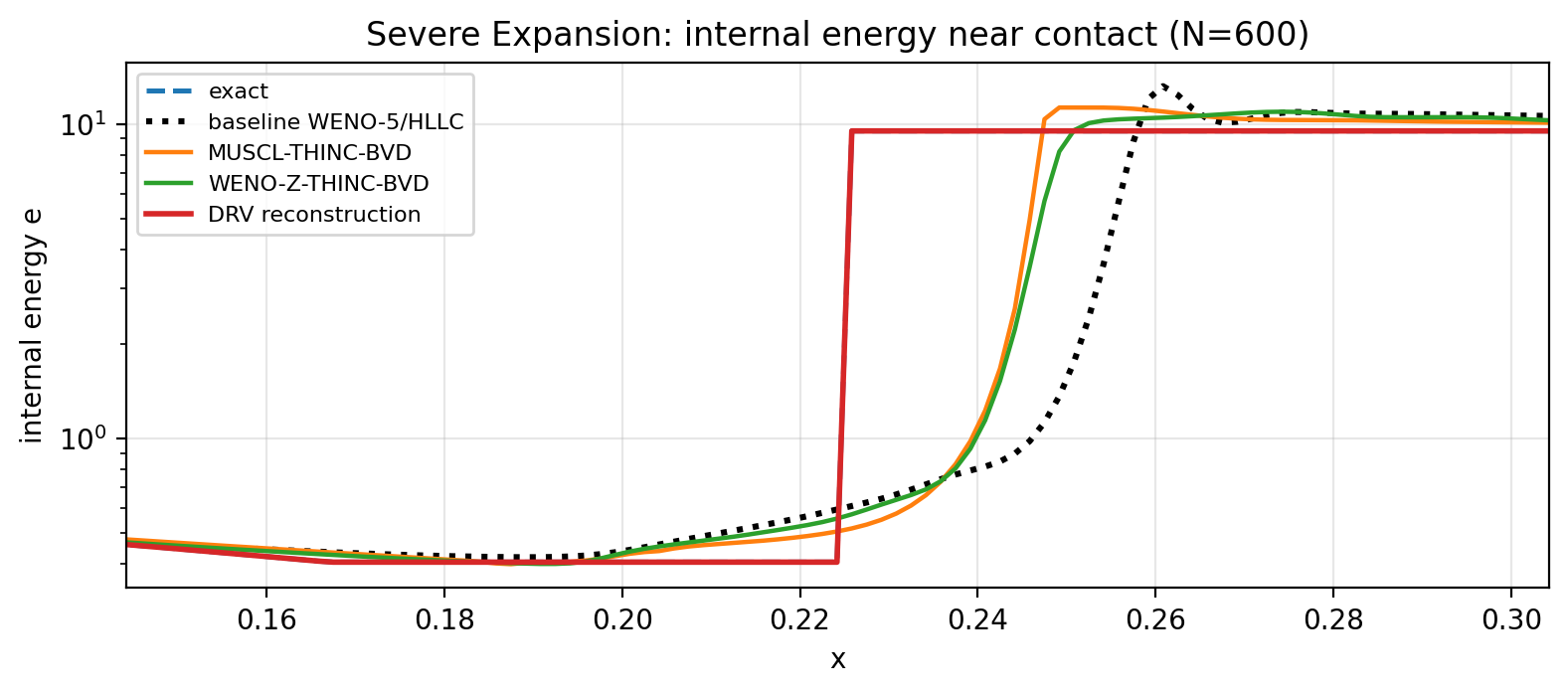}

\medskip
\includegraphics[width=0.88\textwidth]{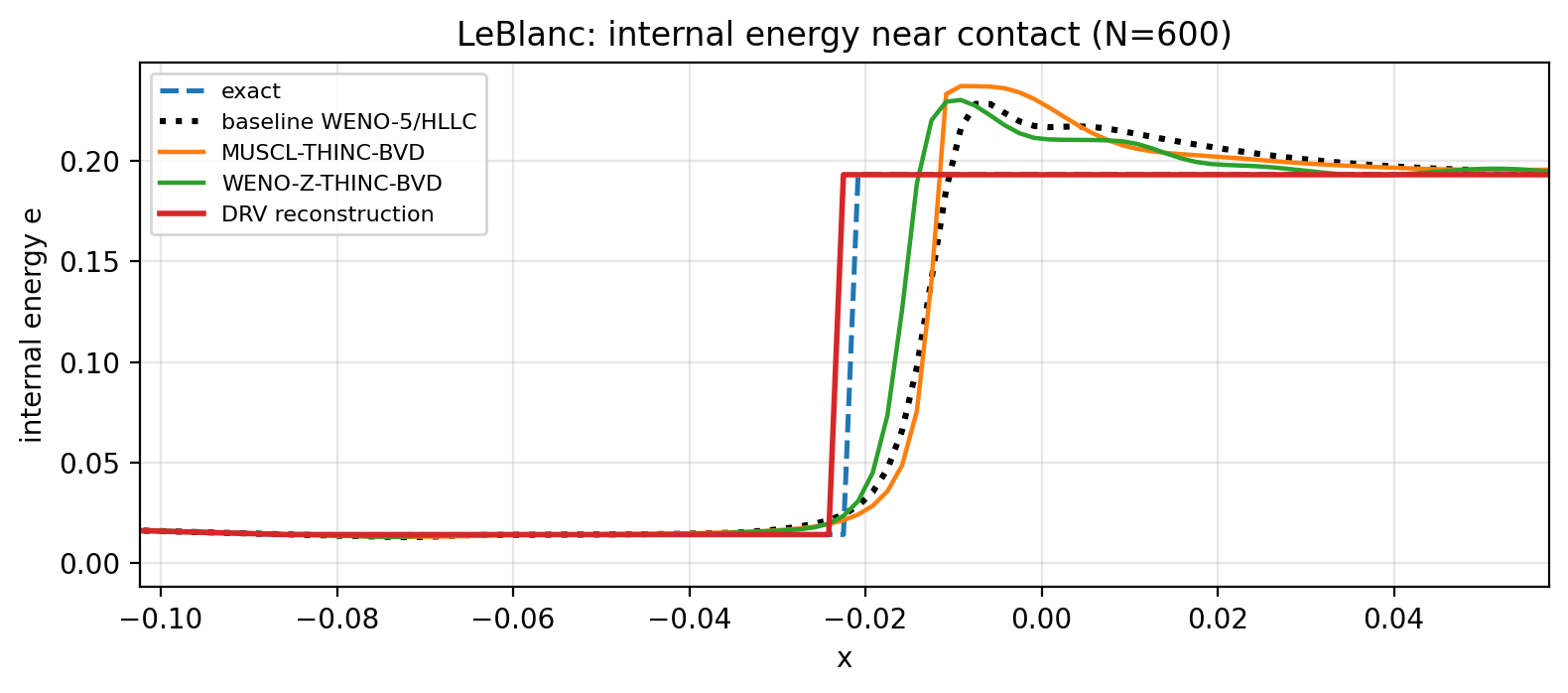}
\caption{Independent comparison with the two jump-like non-DRV reconstructions at common resolution $N=600$ cells. Top: Sod density. Middle: severe-expansion internal energy near the contact. Bottom: LeBlanc internal energy near the contact. At the plotting scales shown here, the one-step and exact-closure DRV curves are visually indistinguishable; in either case the DRV reconstruction is visibly sharper and is the only one among the three reconstructions that eliminates the positive LeBlanc contact overshoot while keeping the contact nearly discontinuous.}
\label{fig:comparison-panels}
\end{figure}

\section{Generalization to all four wave-pattern configurations}
\label{sec:generalized}

The three benchmark problems in Section~\ref{sec:results} all produce the classical $1$-rarefaction/$2$-contact/$3$-shock ordering. Because the three DRV fields $\mrz$, $\mrs$, $\mrw$ separate the three wave families independently of whether each acoustic wave is a rarefaction or a shock, the same localization strategy extends naturally to the three remaining single-interface Riemann configurations:
\begin{center}
\begin{tabular}{ll}
$1$-R\,/\,$2$-C\,/\,$3$-R & (two rarefactions), \\
$1$-S\,/\,$2$-C\,/\,$3$-S & (two shocks), \\
$1$-S\,/\,$2$-C\,/\,$3$-R & (reversed pattern).
\end{tabular}
\end{center}
No change is required in Steps~1--5 of Algorithm~\ref{alg:recon}. Only the closure in Step~6 and the reconstruction in Step~7 need to be made pattern-agnostic.

\subsection{Pattern-agnostic closure}

The main-body one-step closure clips the star pressure to the interval $p_R(1+\epsilon)\le p_* \le p_L(1-\epsilon)$, which enforces the $1$-rarefaction/$3$-shock ordering. To accommodate all four patterns, this clip is removed and the pressure-wave-function equation
\[
F(p) = f_L(p)+f_R(p)+u_R-u_L
\]
is iterated by Newton's method until the relative residual $|F(p_*)|/(|u_R-u_L|+c_L+c_R)$ drops below $10^{-2}$, or until a maximum of~$20$ steps have been reached. Seeding from the geometric mean of the sampled star pressure and the PVRS estimate, the iteration typically converges in $3$--$5$ steps for the standard-pattern benchmarks and in about~$10$ steps for the near-vacuum Toro~123 problem. The convergence guard is essential: if the sampled states do not support a clean Riemann pattern, the closure reports failure and the reconstruction falls back to the baseline solution rather than degrading it. Because $f_L$ and $f_R$ already contain the shock and rarefaction branches, the converged $p_*$ automatically determines the wave types through the signs of $p_*-p_L$ and $p_*-p_R$.

For the reconstruction, a right-side rarefaction fan is interpolated with the same self-similar simple-wave structure used on the left side, while a left-side shock is represented as a sharp jump at the computed shock position.

\subsection{Four additional test problems}
\label{sec:gen-benchmarks}

Table~\ref{tab:gen-benchmarks} lists the four additional test problems. Each is a standard benchmark from the Euler-solver literature.

\begin{table}[t]
\centering
\small
\begin{tabular}{llcccccc}
\toprule
benchmark & pattern & $N$ & $t_{\mathrm{final}}$ & $\gamma$ & $(\rho_L,u_L,p_L)$ & $(\rho_R,u_R,p_R)$ & CFL \\
\midrule
Lax & 1-R/2-C/3-S & $600$ & $0.13$ & $1.4$ & $(0.445,\,0.698,\,3.528)$ & $(0.5,\,0,\,0.571)$ & $0.40$ \\
Toro~123 & 1-R/2-C/3-R & $600$ & $0.15$ & $1.4$ & $(1,\,-2,\,0.4)$ & $(1,\,2,\,0.4)$ & $0.40$ \\
Left Blast & 1-R/2-C/3-S & $800$ & $0.012$ & $1.4$ & $(1,\,0,\,1000)$ & $(1,\,0,\,0.01)$ & $0.30$ \\
Collision & 1-S/2-C/3-S & $800$ & $0.035$ & $1.4$ & $(5.999,\,19.60,\,460.9)$ & $(5.992,\,-6.196,\,46.10)$ & $0.30$ \\
\bottomrule
\end{tabular}
\caption{Four additional test problems covering all four Riemann wave-pattern configurations. Interface locations: $x_*=0$ for Lax, Toro~123, and Left Blast; $x_*=0.4$ for Collision. Domain: $[-\tfrac12,\tfrac12]$ except Collision $[0,1]$.}
\label{tab:gen-benchmarks}
\end{table}

The Lax problem has the same $1$-R/$2$-C/$3$-S ordering as Sod, but with nonzero initial velocity on the left side, which tests that the DRV detection does not depend on $u_L=u_R=0$. Toro~123 is a symmetric two-rarefaction problem that creates a near-vacuum star region ($p_*\approx 0.002$, $\rho_*\approx 0.02$). The Left Blast is the left sub-problem of the Woodward--Colella blast wave, with an initial pressure ratio of~$10^5$. The Collision problem (Toro Test~4) sends two gas masses toward each other, producing two shocks and a contact.

\subsection{Results}

Table~\ref{tab:gen-errors} reports the wave-location errors and Figure~\ref{fig:gen-4panel} shows the reconstructed density for all four problems.

\begin{table}[t]
\centering
\small
\begin{tabular}{llccc}
\toprule
benchmark & wave & exact position & $E^{\mathrm{init}}$ & $E^{\mathrm{ref}}$ \\
\midrule
Lax & 1-R head & $-0.342363$ & $1.68\times 10^{-2}$ & $0$ \\
Lax & 1-R tail & $-0.212771$ & $3.03\times 10^{-2}$ & $2.91\times 10^{-6}$ \\
Lax & contact & $0.198734$ & $1.35\times 10^{-3}$ & $1.53\times 10^{-6}$ \\
Lax & 3-S & $0.322312$ & $1.23\times 10^{-1}$ & $3.63\times 10^{-6}$ \\
\midrule
Toro~123 & 1-R head & $-0.412250$ & $1.69\times 10^{-2}$ & $< 10^{-12}$ \\
Toro~123 & 1-R tail & $-0.052250$ & $2.57\times 10^{-1}$ & $5.82\times 10^{-7}$ \\
Toro~123 & contact & $0.000000$ & $3.08\times 10^{-1}$ & $0$ \\
Toro~123 & 3-R tail & $0.052250$ & $3.58\times 10^{-1}$ & $5.82\times 10^{-7}$ \\
Toro~123 & 3-R head & $0.412250$ & $1.69\times 10^{-2}$ & $< 10^{-12}$ \\
\midrule
Left Blast & 1-R head & $-0.448999$ & $1.41\times 10^{-2}$ & $< 10^{-8}$ \\
Left Blast & 1-R tail & $-0.166796$ & $4.56\times 10^{-1}$ & $2.19\times 10^{-8}$ \\
Left Blast & contact & $0.235169$ & $5.55\times 10^{-2}$ & $1.31\times 10^{-8}$ \\
Left Blast & 3-S & $0.282210$ & $9.67\times 10^{-3}$ & $2.17\times 10^{-8}$ \\
\midrule
Collision & 1-S & $0.427636$ & $1.20\times 10^{-2}$ & $2.38\times 10^{-8}$ \\
Collision & contact & $0.704142$ & $1.32\times 10^{-1}$ & $2.48\times 10^{-9}$ \\
Collision & 3-S & $0.828777$ & $8.10\times 10^{-3}$ & $2.42\times 10^{-8}$ \\
\bottomrule
\end{tabular}
\caption{Wave-location errors for the four additional test problems. The initial geometry comes from the DRV detection stage; the refined geometry is obtained after two passes of the pattern-agnostic closure with adaptive Newton iteration (up to~$20$ steps, early exit on convergence).}
\label{tab:gen-errors}
\end{table}

\begin{figure}[t]
\centering
\includegraphics[width=0.98\textwidth]{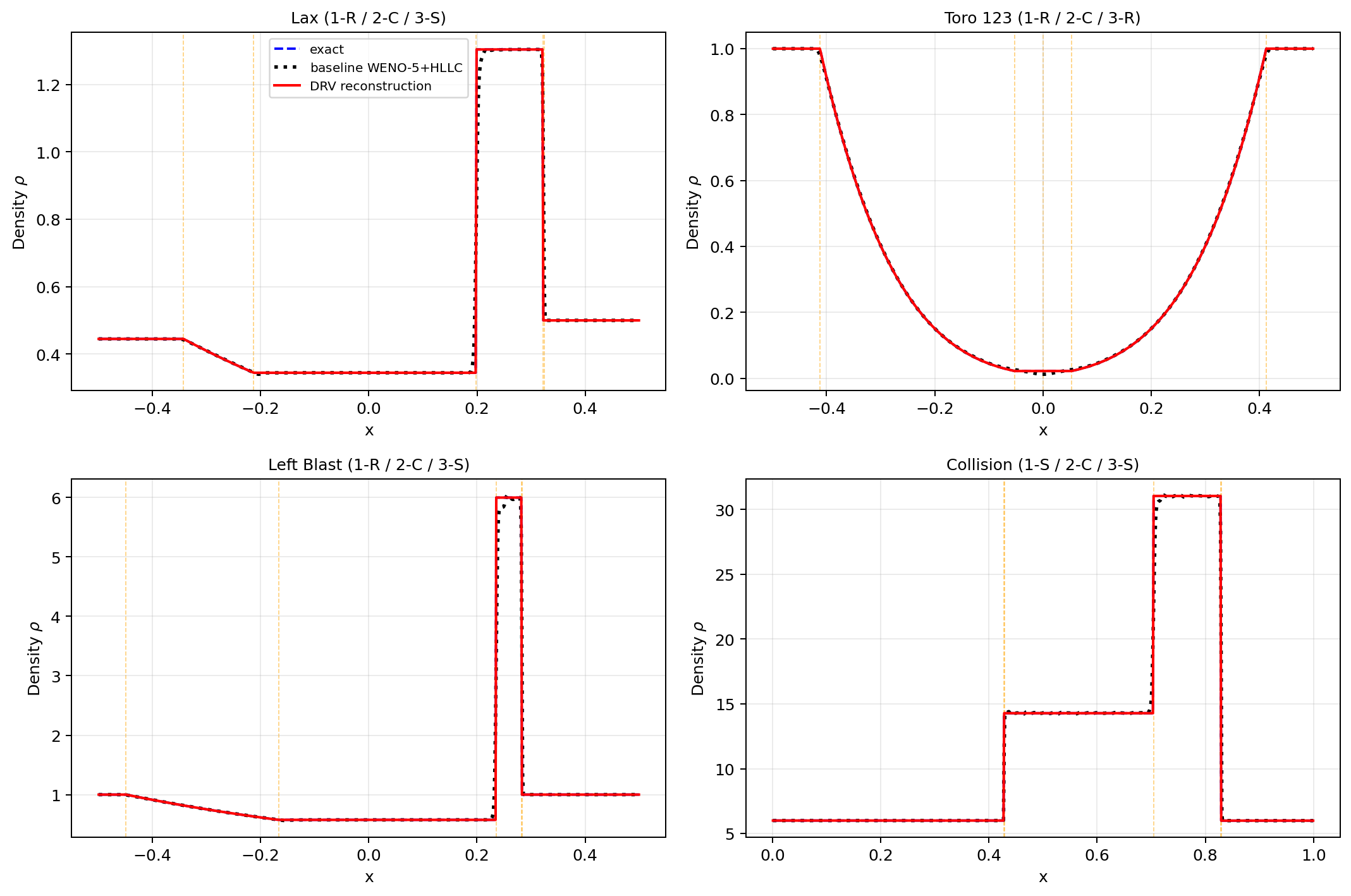}
\caption{Generalized DRV reconstruction on all four wave-pattern configurations. Top left: Lax ($1$-R/$2$-C/$3$-S with nonzero initial velocity). Top right: Toro~123 ($1$-R/$2$-C/$3$-R, near-vacuum star region). Bottom left: Left Blast ($1$-R/$2$-C/$3$-S, pressure ratio~$10^5$). Bottom right: Collision ($1$-S/$2$-C/$3$-S, Toro Test~4). Dashed blue: exact self-similar Riemann solution. Dotted black: baseline WENO--5+HLLC. Solid red: DRV reconstruction.}
\label{fig:gen-4panel}
\end{figure}

\subsubsection{Lax and Left Blast.}
These two problems have the same $1$-R/$2$-C/$3$-S ordering as the main-body benchmarks but test different features. The Lax problem has nonzero initial velocity on the left; the Left Blast has an initial pressure ratio of~$10^5$. Both produce refined wave-location errors at the $10^{-6}$--$10^{-8}$ level and $L^1$ density errors below~$5\times 10^{-6}$, confirming that the generalized closure handles nonzero initial velocities and extreme pressure ratios with no difficulty.

\subsubsection{Collision (two-shock configuration).}
The Collision problem is the first test of the $1$-S/$2$-C/$3$-S pattern. Both shock positions and the contact are placed to $O(10^{-8})$, confirming that the pattern-agnostic closure correctly identifies and locates two shocks with no rarefaction. This is the wave pattern most different from the original $1$-R/$2$-C/$3$-S benchmarks, and the quality of the reconstruction is comparable to the main-body Sod and Left Blast results.

\subsubsection{Toro~123 (near-vacuum two-rarefaction test).}
This is the most challenging test because of the near-vacuum star region ($p_*\approx 0.002$, $\rho_*\approx 0.02$). The sampled star pressures are initially far from the true star pressure, so a single Newton correction is not enough: the closure requires about $10$ Newton steps to converge. The convergence guard plays a critical role here. Without it, a three-step closure produces a large residual ($|F|/\text{scale}\approx 0.19$) and the resulting reconstruction has a density $L^1$ error $30$ times larger than the baseline WENO~solution. The adaptive iteration with the residual-based convergence check prevents this: it either converges the Newton iteration fully or falls back to the baseline. In the present test, the adaptive iteration converges successfully and places all five wave boundaries to $O(10^{-7})$ or better, including the two fan tails that delimit the narrow near-vacuum star region. The reconstructed density tracks the exact self-similar solution through both rarefaction fans and across the near-vacuum center.

\subsection{Computational cost}

\begin{table}[t]
\centering
\small
\begin{tabular}{lrcrrr}
\toprule
benchmark & $N$ & time steps & baseline (s) & DRV post (s) & cost ratio \\
\midrule
Sod & 600 & $\sim$270 & 1.19 & 0.003 & 0.22\% \\
Severe expansion & 1000 & $\sim$400 & 6.16 & 0.004 & 0.06\% \\
LeBlanc & 1000 & $\sim$1100 & 10.74 & 0.004 & 0.04\% \\
Lax & 600 & $\sim$790 & 2.24 & 0.002 & 0.10\% \\
Toro~123 & 600 & $\sim$620 & 1.42 & 0.002 & 0.15\% \\
Left Blast & 800 & $\sim$1200 & 4.50 & 0.002 & 0.05\% \\
Collision & 800 & $\sim$2800 & 7.07 & 0.002 & 0.03\% \\
\bottomrule
\end{tabular}
\caption{Wall-clock cost of the DRV postprocessing relative to the baseline WENO--5/HLLC/SSP--RK3 solve. The DRV stage (algebraic surrogates, adaptive filtering, spike detection, plateau sampling, Newton closure, and piecewise reconstruction) costs between $0.03\%$ and $0.22\%$ of the baseline solve across all seven benchmarks. Timings are medians over multiple runs on a single core.}
\label{tab:cost}
\end{table}

Table~\ref{tab:cost} reports the wall-clock cost of the DRV postprocessing relative to the baseline finite-volume solve across all seven benchmarks. The entire postprocessing pipeline---algebraic DRV surrogates, adaptive Gaussian filtering, spike detection, plateau sampling, Newton closure, and piecewise reconstruction---costs between $0.03\%$ and $0.22\%$ of the baseline solve. In absolute terms, the postprocessing takes $2$--$4$ milliseconds regardless of the problem, because it operates on a single final-time snapshot rather than marching through hundreds or thousands of time steps.

The accuracy improvement is correspondingly large. For the Sod problem, the density $L^1$ error drops by more than ten orders of magnitude; for the Left Blast and Collision problems, by five to six orders of magnitude. Even for the hardest main-body benchmarks (severe expansion and LeBlanc), the improvement factors in the internal-energy $L^1$ error are $54$ and $15$, respectively. The cost--benefit ratio is thus strongly favorable: a negligible computational overhead produces qualitative improvements in wave-location accuracy, contact sharpness, and thermodynamic-defect suppression.

\section{Conclusion}

The main point of this paper is that DRVs are not merely convenient sensors. Their characteristic origin, conservative derivative form, and contact orthogonality are precisely what make them the right variables for spike detection. In particular, the algebraic formulas used in the code are exact smooth DRV identities, not empirical combinations chosen after the fact.

From the numerical viewpoint, the method is a DRV-guided final-time reconstruction for one-dimensional Euler shock tubes. Its accuracy comes from a clear division of labor: filtered DRV surrogates localize the wave geometry in the correct family, plateau sampling extracts the local states, and a very cheap local star-state closure---a few Newton corrections of the pressure-wave-function equation---collapses the remaining geometric error. The entire postprocessing pipeline adds less than $0.25\%$ to the cost of the baseline solve while reducing density errors by factors ranging from~$3$ to more than~$10^{10}$ across the seven benchmarks tested. The wave-location table and the short grid study show that, once the correct pattern is identified, this closure already produces very sharp sub-cell placement of shocks, contacts, and rarefaction edges together with essentially one-cell contact widths. More Newton steps can be taken whenever a more general local wave pattern must be resolved, and Appendix~\ref{app:exact-closure} shows that a full exact local Riemann closure is available within the same DRV framework.

The generalization in Section~\ref{sec:generalized} demonstrates that the DRV mechanism extends naturally to all four single-interface Riemann configurations: removing the pattern-specific pressure clip and iterating the Newton solve to convergence is the only structural change required. All four test problems---including the near-vacuum Toro~123 two-rarefaction test and the Toro~4 two-shock collision---achieve wave-location errors at the $10^{-6}$--$10^{-8}$ level. The convergence guard, which falls back to the baseline whenever the Newton residual is too large, ensures that the reconstruction never degrades the baseline solution.

Placed against the existing reconstruction literature, the broad statement remains modest: sharpening discontinuities is not new. The narrower claim supported here is that the DRV/Radon-spike mechanism gives a practically different level of performance on the benchmark problems studied. Among the two strong jump-like non-DRV comparison schemes independently implemented herein, neither reproduces the combined properties of the DRV method: near-discontinuous contact reconstruction, very small contact-window internal-energy error in the severe-expansion problem, and elimination of the positive LeBlanc contact-layer overshoot.

A natural next step is to export the present local closure from the pure shock-tube setting. One possibility is to apply the pattern-agnostic closure of Section~\ref{sec:generalized} to multi-interface problems in which the individual Riemann fans have not yet interacted, using the DRV spike structure to segment the domain into separate local wave groups. Another is to feed the reconstructed state back into the conservative evolution at selected times, subject to a suitable conservation-correction step, which could test whether the present final-time geometric reconstruction can be upgraded into a dynamically corrective algorithm.

\section*{Acknowledgements}
This work was supported by NSF grant DMS-2307680.

\appendix
\section{Baseline conservative solver}
\label{app:implementation}

This appendix records the baseline WENO--5/HLLC/SSP--RK3 finite-volume solver used in Step~1 of Algorithm~\ref{alg:recon}. The solver is standard and is included only for reproducibility.

\subsection{Primitive-variable WENO--5, HLLC, SSP--RK3}

At every stage of the time integrator, the code converts the conservative variables to primitive variables via
\begin{equation}
\label{eq:prim-recovery}
\rho_j \leftarrow \max(\rho_j,\rho_{\min}),
\qquad
u_j = \tfrac{(\rho u)_j}{\rho_j},
\qquad
p_j = \max\!\left((\gamma-1)\Bigl(E_j-\tfrac12\rho_j u_j^2\Bigr),p_{\min}\right),
\end{equation}
with
\[
\rho_{\min}=10^{-14},
\qquad
p_{\min}=10^{-14}.
\]
The derived variables are
\[
e_j=\tfrac{p_j}{(\gamma-1)\rho_j},
\qquad
c_j=\sqrt{\tfrac{\gamma p_j}{\rho_j}},
\qquad
s_j=\log\!\left(\tfrac{p_j}{\rho_j^{\gamma}}\right).
\]

The code reconstructs the primitive variables $\rho$, $u$, and $p$ independently by WENO--JS of order five, using constant extrapolation over three ghost cells at each boundary. For a scalar grid function $v_j$, the left trace at $x_{j+1/2}$ is
\[
v^-_{j+1/2}=\omega_0 q_0+\omega_1 q_1+\omega_2 q_2,
\]
with candidate polynomials
\begin{align*}
q_0 &= \tfrac13 v_{j-2}-\tfrac76 v_{j-1}+\tfrac{11}{6}v_j,
&
\beta_0 &= \tfrac{13}{12}(v_{j-2}-2v_{j-1}+v_j)^2+\tfrac14(v_{j-2}-4v_{j-1}+3v_j)^2,\\
q_1 &= -\tfrac16 v_{j-1}+\tfrac56 v_j+\tfrac13 v_{j+1},
&
\beta_1 &= \tfrac{13}{12}(v_{j-1}-2v_j+v_{j+1})^2+\tfrac14(v_{j-1}-v_{j+1})^2,\\
q_2 &= \tfrac13 v_j+\tfrac56 v_{j+1}-\tfrac16 v_{j+2},
&
\beta_2 &= \tfrac{13}{12}(v_j-2v_{j+1}+v_{j+2})^2+\tfrac14(3v_j-4v_{j+1}+v_{j+2})^2,
\end{align*}
and nonlinear weights
\[
\alpha_0=\tfrac{0.1}{(\beta_0+\varepsilon_{\mathrm{WENO}})^2},
\qquad
\alpha_1=\tfrac{0.6}{(\beta_1+\varepsilon_{\mathrm{WENO}})^2},
\qquad
\alpha_2=\tfrac{0.3}{(\beta_2+\varepsilon_{\mathrm{WENO}})^2},
\qquad
\omega_k=\tfrac{\alpha_k}{\alpha_0+\alpha_1+\alpha_2},
\]
with $\varepsilon_{\mathrm{WENO}}=10^{-6}$. The right trace $v^+_{j+1/2}$ is obtained by the reflected formulas used in the classical WENO--JS implementation. These formulas are applied componentwise to $\rho$, $u$, and $p$.

Given interface states $(\rho_L,u_L,p_L)$ and $(\rho_R,u_R,p_R)$, the code uses the Davis/Einfeldt estimates
\[
S_L=\min(u_L-c_L,u_R-c_R),
\qquad
S_R=\max(u_L+c_L,u_R+c_R),
\]
where $c_{L/R}=\sqrt{\gamma p_{L/R}/\rho_{L/R}}$. The auxiliary HLL flux is
\[
F_{\mathrm{HLL}}
=
\tfrac{S_R F(U_L)-S_L F(U_R)+S_LS_R(U_R-U_L)}{\clip(S_R-S_L)},
\]
with the sign-preserving denominator regularization
\[
\clip(a)=\operatorname{sign}(a)\max(|a|,10^{-14}).
\]
The HLLC contact speed is \cite{ToroSpruceSpeares1994}
\begin{equation}
\label{eq:hllc-sm}
S_M=
\tfrac{p_R-p_L+\rho_Lu_L(S_L-u_L)-\rho_Ru_R(S_R-u_R)}{\clip\bigl(\rho_L(S_L-u_L)-\rho_R(S_R-u_R)\bigr)}.
\end{equation}
The star densities are
\[
\rho_{*L}=\rho_L\tfrac{S_L-u_L}{\clip(S_L-S_M)},
\qquad
\rho_{*R}=\rho_R\tfrac{S_R-u_R}{\clip(S_R-S_M)},
\]
and the star energies are
\begin{align*}
E_{*L} &= \rho_{*L}\!\left(\tfrac{E_L}{\rho_L}+(S_M-u_L)\left(S_M+\tfrac{p_L}{\clip(\rho_L(S_L-u_L))}\right)\right),\\
E_{*R} &= \rho_{*R}\!\left(\tfrac{E_R}{\rho_R}+(S_M-u_R)\left(S_M+\tfrac{p_R}{\clip(\rho_R(S_R-u_R))}\right)\right),
\end{align*}
with $E_k=p_k/(\gamma-1)+\tfrac12\rho_k u_k^2$. The HLLC flux is then
\[
F_{\mathrm{HLLC}}=
\begin{cases}
F(U_L), & S_L\ge 0,\\
F(U_L)+S_L(U_{*L}-U_L), & S_L<0\le S_M,\\
F(U_R)+S_R(U_{*R}-U_R), & S_M<0<S_R,\\
F(U_R), & S_R\le 0.
\end{cases}
\]
The code checks that $S_M$, $\rho_{*L}$, $\rho_{*R}$, $E_{*L}$, $E_{*R}$, and the two star pressures are finite and positive; whenever that validity test fails, the interface flux falls back to $F_{\mathrm{HLL}}$. This is exactly what the accompanying Python code does.

The semi-discrete update is
\[
\tfrac{d}{dt}U_j = -\tfrac{F_{j+1/2}-F_{j-1/2}}{\Delta x}.
\]
After the flux difference is assembled, the first and last cell updates are set to zero. Combined with constant extrapolation in the reconstruction, this is the edge treatment used in the code. Time stepping is third-order SSP--RK3:
\begin{align*}
U^{(1)} &= U^n + \Delta t\,L(U^n),\\
U^{(2)} &= \tfrac34 U^n + \tfrac14\bigl(U^{(1)}+\Delta t\,L(U^{(1)})\bigr),\\
U^{n+1} &= \tfrac13 U^n + \tfrac23\bigl(U^{(2)}+\Delta t\,L(U^{(2)})\bigr),
\end{align*}
with
\[
\Delta t = \mathrm{CFL}\,\tfrac{\Delta x}{\max_j(|u_j|+c_j)+10^{-14}}.
\]
The benchmark-dependent CFL numbers are the ones hard-coded in the Python file and listed in Section~\ref{sec:results}.

\section{Exact local Riemann closure as an optional sharpening module}
\label{app:exact-closure}

The DRV separation makes the exact local Riemann closure available for the same reason that it makes the one-step closure available: once the filtered spikes have isolated the $1$-wave, the contact, and the $3$-wave, the plateau windows determine sampled left and right states together with two sampled star-side states. One may therefore solve a new local ideal-gas Riemann problem built from those sampled states. This is not the original benchmark solution being reused. The sampled plateau velocities need not vanish, and the local states are extracted \emph{a posteriori} from the computed numerical profile.

In the exact-closure variant, the same pressure equation
\[
F(p)=f_L(p)+f_R(p)+u_R-u_L
\]
is iterated to convergence by Newton's method, using the same pressure functions $f_L$ and $f_R$ as in Section~\ref{sec:strategy}. The code supplements this with the classical vacuum check
\[
\tfrac{2(c_L+c_R)}{\gamma-1} \le u_R-u_L,
\]
uses the PVRS pressure as an auxiliary fallback seed, and iterates until the relative pressure change is below $10^{-10}$ or until $50$ iterations have been reached. Whenever the returned pattern is the expected $1$-rarefaction/$2$-contact/$3$-shock configuration, the exact star states and exact wave speeds replace the sampled ones.

Table~\ref{tab:exact-vs-onestep} compares the main-body one-step closure with this fully converged exact local Riemann closure at the benchmark resolutions. For Sod and LeBlanc, the two are effectively indistinguishable to the displayed digits. For the severe-expansion problem, the exact closure sharpens the last few $10^{-4}$--$10^{-3}$ units of wave-location error and drives the contact-window internal-energy error even closer to zero, but the one-step closure already produces the same qualitative reconstruction and the same suppression of a positive contact overshoot.

\begin{table}[t]
\centering
\small
\setlength{\tabcolsep}{4.5pt}
\begin{tabular}{llcccccc}
\toprule
& & $W_\rho$ & $E_{e,\mathrm{ct}}$ & $O_e^+$ & $E_{rt}$ & $E_c$ & $E_s$ \\
\midrule
\multirow{2}{*}{Sod}
 & 1-step & $1.33\!\times\! 10^{-3}$ & $5.26\!\times\! 10^{-14}$ & none & $2.60\!\times\! 10^{-15}$ & $1.13\!\times\! 10^{-14}$ & $8.11\!\times\! 10^{-14}$ \\
 & exact  & $1.33\!\times\! 10^{-3}$ & $0$ & none & $0$ & $0$ & $0$ \\
\midrule
\multirow{2}{*}{\shortstack[l]{Severe\\expansion}}
 & 1-step & $8.01\!\times\! 10^{-4}$ & $6.00\!\times\! 10^{-4}$ & none & $1.90\!\times\! 10^{-4}$ & $2.07\!\times\! 10^{-4}$ & $5.59\!\times\! 10^{-4}$ \\
 & exact  & $8.00\!\times\! 10^{-4}$ & $<\!10^{-15}$ & none & $0$ & $0$ & $<\!10^{-15}$ \\
\midrule
\multirow{2}{*}{LeBlanc}
 & 1-step & $8.01\!\times\! 10^{-4}$ & $1.88\!\times\! 10^{-4}$ & none & $1.69\!\times\! 10^{-4}$ & $1.84\!\times\! 10^{-4}$ & $2.46\!\times\! 10^{-4}$ \\
 & exact  & $8.01\!\times\! 10^{-4}$ & $1.88\!\times\! 10^{-4}$ & none & $1.69\!\times\! 10^{-4}$ & $1.84\!\times\! 10^{-4}$ & $2.46\!\times\! 10^{-4}$ \\
\bottomrule
\end{tabular}
\caption{One-step sampled-pressure closure (``1-step'') versus exact local Riemann closure (``exact'') at the benchmark resolutions. Here $W_\rho$ is the density contact width, $E_{e,\mathrm{ct}}$ is the contact-window internal-energy $L^1$ error, $O_e^+$ is the positive contact-window overshoot (``none'' = no positive overshoot detected), and $E_{rt}$, $E_c$, $E_s$ are the refined fan-tail, contact, and shock location errors. The exact closure is optional: for the present benchmarks it mainly removes the last residual severe-expansion error, while Sod and LeBlanc are already essentially saturated by the one-step version.}
\label{tab:exact-vs-onestep}
\end{table}

\end{document}